\documentclass{article}

\usepackage{hyperref}
\usepackage{fullpage}
\usepackage{amsmath}
\usepackage{amssymb}
\usepackage{amsthm}
\usepackage{enumitem,xspace}
\usepackage{float}
\usepackage{tikz}
\usetikzlibrary{arrows.meta}
\usepackage{url}
\usepackage{tabularx}
\usepackage{makecell}
\usepackage{cleveref}
\usepackage{subcaption}

\newtheorem{theorem}{Theorem}[section]
\newtheorem{definition}[theorem]{Definition}

\newtheorem{lemma}[theorem]{Lemma}

\newtheorem{claim}[theorem]{Claim}

\newcommand{\farmersmarket}{{street fair}\xspace}

\newcommand{\mc}[1]{\mathcal{#1}}
\newcommand{\wh}{\widehat}
\newcommand{\p}{\partial}

\begin{document}

\title{Why is My Route Different Today? \\An Algorithm for Explaining Route Selection\footnote{A preliminary version of this paper appeared in the SIAM Conference on Applied and Computational Discrete Algorithms (ACDA) 2025. Code for the experiments can be found here: \url{https://github.com/google-research/google-research/tree/master/explainable_routing}}}

\author{Aaron Schild\thanks{Google Research, Mountain View, CA}
\and Sreenivas Gollapudi\thanks{Google Research, Mountain View, CA}
\and Anupam Gupta\thanks{New York University, New York, NY}
\and Kostas Kollias\thanks{Google Research, Mountain View, CA}
\and Ali Sinop\thanks{Google Research, Mountain View, CA}}

\maketitle

\begin{abstract}
Users of routing services like Apple Maps, Google Maps, and Waze frequently wonder why a given route is proposed. This question particularly arises when dynamic conditions like traffic and road closures cause unusual routes to be proposed. While many dynamic conditions may exist in a road network at any time, only a small fraction of those conditions are typically relevant to a given user's route. In this work, we introduce the concept of an \emph{simple valid explanation} (SVE), which consists of a small set of traffic-laden road segments that answer the following question: \emph{Which traffic conditions cause a particular shortest traffic-aware route to differ from the shortest traffic-free route?} We give an efficient algorithm for finding SVEs and show that they theoretically and experimentally lead to small and interpretable answers to the question.
\end{abstract}

\section{Introduction}\label{sec:intro}

Optimization and machine learning often result in solutions that are difficult to interpret. The machine learning and optimization communities have developed numerous techniques for explaining the behavior of machine learning models~\cite{ll17,wotm24}. These techniques are often insufficient in explaining the output of complex models like neural networks. In those applications, simpler models like decision trees must be used instead
~\cite{LundbergECDPNKH20,dfmr20}.
However, simple models can often lead to optimality loss on problems like clustering~\cite{gpsy23}.

In this paper, we 
study explainability in graph algorithms, particularly in route recommendations. 
People interact with the output of graph algorithms on a daily basis, particularly in the context of routing engines like Apple Maps, Google Maps, and Waze. However, it can be difficult to understand why these engines propose a particular route. People are much more adept at assessing the meaning of text, pictures, or video, for instance, than they are at assessing the optimality of a route given to them by a routing engine. 
Indeed, the optimality of a route can hinge on global information. This complexity manifests itself particularly when routes change. For example, a user may be given an unusual route through a small town due to the presence of a \farmersmarket that they were unaware of 20 blocks away. In this case, the user will not notice the \farmersmarket in their viewport. The user may choose to not follow the routing engine's path because of this, eventually leading them to a complicated drive around the \farmersmarket. To prevent this from happening, the routing engine should do the following:

\begin{enumerate}
    \item identify the relevance of the \farmersmarket to its unusual choice of route, and 
    \item present this information in a succinct and interpretable way to the user.
\end{enumerate}

In this work, we focus on the first part: \emph{identification}. The second aspect can be achieved either by simply listing out the names of streets returned from the identification step, or by using identification as a retrieval step in a retrieval-augmented-generation (RAG)-based large language model~\cite{lep20,ram23,setal24}.

A strawman solution for the  identification step is to report all obstructions along the default no-traffic route. 
While such an approach may work if the routing engine were returning the second-best path to the user, it would not work if several natural options were blocked. Furthermore, even in the case where just the best path were blocked, not all obstructions along that path may be relevant to the routing engine's choice of route. Thus, a more general and nuanced solution is desirable.

We instead approach the identification step by formulating it as an optimization problem. To do this, 
we first need to formalize the notion of an explanation. In the \farmersmarket example, the \farmersmarket is represented in a road network by a collection of road segments. Many road segments face 
unusual circumstances during the \farmersmarket---a few segments are closed, and many others have higher traffic than normal. Formally, this information is represented by changing weights of arcs in the road network, where arc weights represent travel times to cross the corresponding road segment. (The closed segments have infinite arc weight, while the others have higher arc weight than normal.) The crux of the identification problem is this: while \textit{many} segments will have higher weight than normal, the only weights that are relevant are the ones near the \farmersmarket.

\subsection{First Attempt At Explanations: Subselecting Traffic}

We now work towards a definition of ``explanation'' in our context. In the \farmersmarket example, the weights on the segments near the \farmersmarket are high enough to make it so that the shortest origin-destination path \textit{is} the path presented to the user. If the weights of the segments adjacent to the \farmersmarket were replaced with their free-flow values, the shortest origin-destination path would have instead gone through the \farmersmarket. Formally, let $G$ be the directed graph representing the road network. The arcs $e$ of $G$ each have two different weights: $\ell(e)$, the free-flow travel time required to traverse $e$, and $u(e)$, the travel time require to traverse $e$ with current traffic.\footnote{We assume that $u(e)$ is not a function of time, unlike traffic in real-world road networks, for simplicity. Time dependence can be modeled by considering a layered graph where each vertex represents a place-time pair, as is standard.} Naturally, it makes sense to require that $\ell(e) \le u(e)$ for all arcs $e$ in $G$, since traffic only increases travel times.\footnote{Not that such changes are impossible, but if they do happen, it is because of more significant structural changes in the underlying road network.} This arc weight model also captures road closures -- simply make $u(e)$ a very high number if $e$ is a closed road segment. Fix two vertices $s$ and $t$ -- the origin and destination for the desired path. Let $P$ the path returned to the user by the routing engine, routing the user from vertex $s$ to vertex $t$. In particular, $P$ is the shortest $s$-$t$ path in $G$ weighted by the weights $u(e)$. An \emph{explanation} of the path $P$ should be a set $X\subseteq E(G)$ of arcs for which if we keep the weight $u(e)$ for all $e\in X$ and decrease the weight for any $e\notin X$ down to $\ell(e)$, $P$ remains the shortest $s$-$t$ path in $G$.

There are many possible explanations $X$ of a path $P$ in general. In line with Occam's Razor, it is natural to ask for the \emph{simplest} explanation, i.e., an explanation $X$ for which $|X|$ is minimized. In the \farmersmarket example, the explanation should only include segments that are close to the street fair, as traffic on faraway segments is not relevant to the choice of $P$. We will demonstrate both experimentally and theoretically that the simplest explanation is a good one.

\subsection{Formulating an LP Relaxation for Explanations and Our Contributions}

Unfortunately, the explanations that we discussed in the previous section generalize the problem of shortest path interdiction and are thus NP-hard to find \cite{BGV89} and are hard to even approximately find assuming the Unique Games Conjecture \cite{L17}. The main conceptual contribution of our paper, the concept of a \emph{simple valid explanation}, is a carefully-chosen linear programming relaxation of this NP-hard optimization problem. 

We start by formulating the problem of finding the desired explanation as a mixed-integer linear program, which has variables $w_e$ for all $e\in E(G)$ and $d_v$ for all $v\in V(G)$:

\begin{align}
    \min \sum_{e \in E(G)} & \mathbf{1}[w_e \ne \ell(e)] && \tag{MIP} \label{MIP}\\
    w_e &\in \{\ell(e), u(e)\} && \forall e \in E(G) \label{mip:1}\\
    d_v - d_u &\leq w_e & &\forall e = (u,v) \in E(G) \label{mip:2} \\
    d_v - d_u &= w_e && \forall e = (u,v) \in P. \label{mip:3}
\end{align}

Constraints \ref{mip:2} and \ref{mip:3} enforce the constraint of $P$ remaining the shortest $s$-$t$ path. Constraint \ref{mip:1} ensures that each arc has one of the two allowed arc weights, while the objective counts the number of arcs for which traffic is used. For this section only, define $\tau(e) = 0$ when $u(e) = \ell(e)$ and $\tau(e) = \frac{1}{u(e) - \ell(e)}$ otherwise. We can rewrite each arc's $e$ contribution to the objective as follows when $u(e) \ne \ell(e)$\footnote{$e$ does not contribute to the objective when $u(e) = \ell(e)$.}:

$$\mathbf{1}[w_e \ne \ell(e)] = \frac{w_e - \ell(e)}{u(e) - \ell(e)} = \tau(e) (w_e - \ell(e))$$

This substitution leads to a natural relaxation:

\begin{align}
    \min \sum_{e \in E(G)} & \tau(e)(w_e - \ell(e)) && \tag{LP1} \label{LP1a}\\
    \ell(e) \le w_e &\le u(e) && \forall e \in E(G) \label{lp:1}\\
    d_v - d_u &\leq w_e & &\forall e = (u,v) \in E(G) \label{lp:2} \\
    d_v - d_u &= w_e && \forall e = (u,v) \in P. \label{lp:3}
\end{align}

Conceptually, we can state the constraints of this linear program, which we call \emph{sufficiency} and \emph{validity}, as follows:

\begin{enumerate}
    \item (Sufficiency, constraints \ref{lp:2} and \ref{lp:3}) $P$ is a shortest $s$-$t$ path in $G$ under the arc weights $w_e$.
    \item (Validity, constraint \ref{lp:1}) $\ell(e) \le w_e \le u(e)$ for all arcs $e$ in $G$.
\end{enumerate}

An optimizer $w$ of the linear program is called \emph{$\tau$-simple} due to its minimization of the objective. $w$ is called a \emph{($\tau$-)simple valid explanation (SVE)}. While $w$ can be computed using any linear program solver, such solvers are not efficient on continent-scale road networks, which contain hundreds of millions of arcs. Instead, in Section \ref{sec:algo}, we give a novel flow-based algorithm for solving the linear program. It works well for integral values of $\tau(e)$, so we use different choices of $\tau$ from the one used in this section in our experiments. This change of $\tau$ does not meaningfully impact our experimental results. In Section \ref{sec:algo}, we also discuss the considerations behind the choice of $\tau$.

One of the key contributions of this paper is showing that the simplicity criterion results in SVEs being small in practice. In Section \ref{sec:closures}, we define our evaluation setup and give some theoretical results validating the succinctness of SVEs. We give a broad set of examples in Section \ref{sec:examples} that illustrate how SVEs give intuitive explanations for routing decisions. In Section \ref{sec:experiments}, we present experimental results which show that SVEs are short in many settings. 

\subsection{Preliminaries and the Formal Definition of SVEs}

We are now ready to formally repeat the definition of SVEs given in the previous section. In the following, 
consider a directed graph $G$. Let $E(G)$ denote the arc (edge) set of $G$ and let $V(G)$ denote the vertex set of $G$. Consider two vertices $s,t\in V(G)$ and a path $P\subseteq E(G)$ from $s$ to $t$. Let $\ell,u\in \mathbb{R}_{\ge 0}^{E(G)}$ denote two different sets of weights for the arcs in $G$, with $\ell(e) \le u(e)$ for all $e\in E(G)$. 

\begin{definition}[Simple Valid Explanations] An \emph{explanation} of the path $P$ is a collection of weights $w_e\in \mathbb{R}_{\ge 0}$ for all $e\in E(G)$ with the following property:
\begin{enumerate}
    \item (Sufficiency) $P$ is a shortest $s$-$t$ path in the digraph $G$ with arc weights $\{w_e\}_{e\in E(G)}$.
\end{enumerate}
We call an explanation of $P$ \emph{valid} if it satisfies the following constraint:
\begin{enumerate}[start=2]
    \item (Validity) $\ell(e) \le w_e \le u(e)$ for all $e\in E(G)$.
\end{enumerate}
Let $\mc W(\ell,u,P)$ denote the set of valid explanations for $P$. Call an arc $e\in E(G)$ \emph{pliable} if $\ell(e) \ne u(e)$; i.e. $\ell(e) < u(e)$. Let $X(\ell,u)\subseteq E(G)$ denote the set of pliable arcs in $G$.
Let $\tau\in \mathbb{R}^{E(G)}_{\ge 0}$ be a vector with nonnegative coordinates. The $\tau$-\emph{valuation} of a pliable arc $e$ with respect to a valid explanation $w$ is defined as $V_{\tau}(w,e) := \tau(e)(w_e - \ell(e)).$ The $\tau$-\emph{valuation} of a valid explanation $w$ is defined as the sum of the valuations
of its arcs, 
$V_{\tau}(w) := \sum_{e\in X(\ell,u)} V_{\tau}(w,e).$
A valid explanation $w\in \mc W(\ell,u,P)$ is called $\tau$-\emph{simple} if it satisfies: 
\begin{enumerate}[start=3]
    \item (($\tau$-)Simplicity) the explanation $w$ has minimum $\tau$-valuation; i.e.,
    $V_{\tau}(w) = \min_{w'\in \mc W(\ell,u,P)} V_{\tau}(w')$.
\end{enumerate}
For a valid explanation $w$, let $\text{support}(w)$ denote the subset of $e\in E(G)$ for which $w_e > \ell(e)$. We call a valid explanation $w$ \emph{nontrivial} if $\text{support}(w)$ is nonempty. We drop $\tau$ from these definitions when its value is clear from context.
\end{definition}

\subsection{Illustrating the Simplicity Criterion}

We now illustrate the meaning of the simplicity criterion in terms of its discriminative power between different valid explanations. Fix a number $k > 1$ and form a non-simple digraph $G$ as follows. $G$ has three vertices: $s$, $v$, and $t$. There are $k$ parallel arcs $e_i = (s,v)$, making $G$ a non-simple graph. There is a single arc $f = (v,t)$, and a single arc $e = (s,t)$. Thus, $V(G) = \{s,v,t\}$ and $E(G) = \{e_1,\hdots,e_k,f,e\}$. Let $P = \{e\}$. We define the functions $\ell$ and $u$ as follows:

\begin{enumerate}
    \item $\ell(e) = 100$, $\ell(e_i) = 49$ for all $i\in \{1,2,\hdots,k\}$, and $\ell(f) = 49$;
    \item $u(e) = 100$, $u(e_i) = 51$ for all $i\in \{1,2,\hdots,k\}$, and $u(f) = 51$.
\end{enumerate}

Let $\tau(e) = 1$ for every arc $e$. Note that $P$ is not the shortest path under the arc weights $\ell$, so every valid explanation for $P$ is nontrivial. However, there are multiple valid explanations of $P$. For example:
\begin{enumerate}
    \item $x$, which is defined as $x_e = 100$, $x_{e_i} = 51$ for all $i\in \{1,2,\hdots,k\}$, and $x_f = 49$.
    \item $w$, which is defined as $w_e = 100$, $w_{e_i} = 49$ for all $i\in \{1,2,\hdots,k\}$, and $w_f = 51$.
\end{enumerate}
These two explanations are depicted in Figure \ref{fig:two-graphs}. While both are valid, $w$ is clearly the better explanation, as it is more succinct. Instead of penalizing all of the $e_i$'s, one merely has to penalize $f$ to make $P$ the shortest path. The minimum valuation of any valid explanation in this example is $V(w) = 2$. As a result, $w$ is an SVE, while $x$ is only a valid explanation because $V(x) = 2k > 2$. Thus, in this example, simplicity favors the better explanation. This is particularly relevant in the context of bridges, as we illustrate in the next section.

\definecolor{darkgreen}{rgb}{0.0, 0.5, 0.0}

\begin{figure}[H]
    \centering
    \scalebox{0.9}{
        \centering
        \begin{tikzpicture}[
            vertex/.style={circle,draw,minimum size=20pt,inner sep=0pt},
            thinedge/.style={->,>=stealth,thin},
            veryboldedge/.style={->,>=stealth,line width=1.5pt}
        ]
            
            \node[vertex] (s) at (0,0) {s};
            \node[vertex] (v) at (2,0) {v};
            \node[vertex] (t) at (4,0) {t};
            
            \draw[red][veryboldedge] (s) to[bend left=30] (v);
            \draw[red][veryboldedge] (s) to[bend left=0] (v);
            \draw[red][veryboldedge] (s) to[bend right=30] node[midway,below] {51} (v);
            
            \draw[thinedge] (v) to node[midway,below] {49} (t);
            
            \draw[darkgreen][thinedge] (s) to[bend left=40] node[midway,above] {100} (t);
        \end{tikzpicture}
        \centering
        \begin{tikzpicture}[
            vertex/.style={circle,draw,minimum size=20pt,inner sep=0pt},
            thinedge/.style={->,>=stealth,thin},
            veryboldedge/.style={->,>=stealth,line width=1.5pt}
        ]
            
            \node[vertex] (s) at (0,0) {s};
            \node[vertex] (v) at (2,0) {v};
            \node[vertex] (t) at (4,0) {t};
            
            \draw[thinedge] (s) to[bend left=30] (v);
            \draw[thinedge] (s) to[bend left=0] (v);
            \draw[thinedge] (s) to[bend right=30] node[midway,below] {49} (v);
            
            \draw[red][veryboldedge] (v) to node[midway,below] {51} (t);
            
            \draw[darkgreen][thinedge] (s) to[bend left=40] node[midway,above] {100} (t);
        \end{tikzpicture}
    }
    \caption{The valid explanations $x$ (left) and $w$ (right). Arcs $g$ with weights not equal to $\ell(g)$ are red bolded for the green path (just one arc). Note that the total weight of arcs in $w$ is less than in $x$, which matches the fact that $w$ is a $\tau$-SVE while $x$ is not. Thus, the simplicity criterion correctly discriminates between these two explanations by picking the one with fewer edges.}
    \label{fig:two-graphs}
\end{figure}
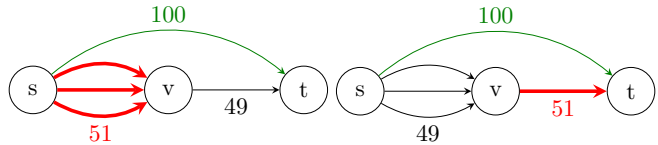

\subsection{Road Network Example}

Consider a real-life example from Seattle, Washington as depicted in Figure~\ref{fig:seattle_bridge_ex}. Suppose that there is heavy traffic on the shortest path, which results in the routing engine suggesting the lower path (along I-90 shown in green). Note that the lower path is considerably longer than the shortest path. The simplest explanation is to assign all the weight to the segments corresponding to the bridge along the upper path (red arc). This results in an intuitive explanation: ``Due to heavy traffic on the Route 520 bridge, take I-90 instead.'' 

\begin{figure}[H]
    \centering
    \includegraphics[width=0.9\columnwidth]{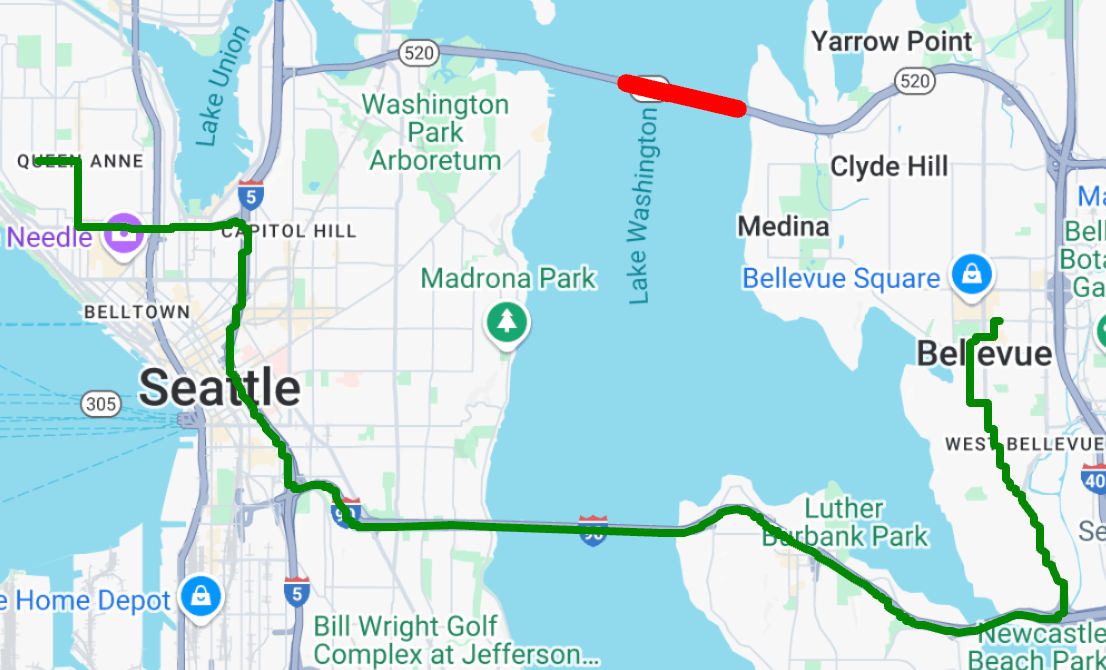}
    \caption{Simplicity criterion on an actual road network. The green route is the shortest path with traffic, with relevant traffic (the SVE) depicted in red. The SVE is the intuitive explanation here: the green path is optimal because of traffic on the 520 bridge, which would be a shorter route under no-traffic conditions.}
    \label{fig:seattle_bridge_ex}
\end{figure}

\subsection{Penalty-Based Explanations: A Natural Baseline}

There is a natural baseline algorithm for constructing an explanation based on the classic penalty method for computing alternate routes \cite{CBB07}. To the best of our knowledge, this is the only known baseline for this problem. This method generalizes the natural concept of taking all obstructions that exist along the shortest path. Given an $s$-$t$ path $P$ to explain with lower and upper bounds $\ell$ and $u$, construct an explanation $w$ for $P$ as follows:

\begin{enumerate}
    \item Initialize $w = \ell$ and $Q = \emptyset$.
    \item While $w(Q) < w(P)$,
    \begin{enumerate}
        \item For every $e\in Q\setminus P$, change $w_e = u(e)$.
        \item Let $Q$ be the shortest $s$-$t$ $w$-weighted path in $G$
    \end{enumerate}
    \item Return $w$.
\end{enumerate}

\begin{lemma}\label{lem:pbe}
If there exists a valid explanation for $P$, then this algorithm terminates in a valid (but not necessarily simple) explanation of $P$; i.e. one that satisfies the validity and sufficiency properties.
\end{lemma}

\begin{proof}
It suffices to show that this algorithm terminates, as when it does, $P$ is the shortest path according to the returned $w$ by definition, and $w$ is valid by construction. Let $x$ be a valid explanation for $P$, which exists by assumption. Without loss of generality, we may assume that $x(e) = \ell(e)$ for all $e\in P$, as $P$ remains a shortest path after decreasing weights on $P$ only. Let $y\in \mathbb{R}^{E(G)}$ denote the vector obtained by letting $y(e) = \ell(e)$ for all $e\in P$ and $y(e) = u(e)$ otherwise. Since increasing weights only increases the length of the shortest path, $y$ must also be a valid explanation for $P$.

Consider two explanations $w$ and $w'$ at the beginning and end of a single iteration respectively, with $Q$ denoting the shortest path from the previous iteration.\footnote{Previous is chosen because $Q$ is updated after $w$ during each iteration.} If $w = w'$, then $w_e = u(e)$ for all $e\in Q\setminus P$. This means that $w(Q) = y(Q)$. Since $y$ satisfies sufficiency, $y(Q) \ge y(P)$. By definition of $y$, $y(P) = \ell(P) = w(P)$. Thus, $w(Q) \ge w(P)$, which means that the while loop should have exited, a contradiction. Thus, $w' \ne w$. Since each iteration adds an additional coordinate $e$ in which $w_e = u(e)$, there can only be $|E(G)|$ iterations, as desired.
\end{proof}

We call the $w$ returned by this algorithm the \emph{penalty-based explanation} (PBE) for $P$. This algorithm just repeatedly takes all obstructions along the shortest path and recalculates the path until the path becomes the desired path. This generalizes the idea of explaining a path by reporting all traffic conditions along the shortest path.

Throughout this work, we compare SVEs to PBEs to get a better understanding for the size of SVEs. There are some important qualitative differences between SVEs and PBEs:

\begin{enumerate}
    \item PBEs contain too many segments. Specifically, PBEs take all obstructions along a shortest path, while SVEs have the potential to select only important obstructions. One could modify the PBE algorithm to only increase some weights along a path rather than all weights, but this adds additional complexity that does not exist in the SVE algorithm.
    \item The PBE algorithm is too path-oriented. In particular, throughout our experiments, we assess SVEs by their ability to recover or subselect obstructions along paths generated by a PBE-like algorithm. PBEs, by construction, will recover all of these obstructions. However, their ability to do this is very brittle: if the PBE algorithm is modified even slightly, it could produce a completely different set of paths and thus include obstructions from other parts of the graph in its explanation. SVEs will not do this.
\end{enumerate}

Our second point highlights a crucial advantage to SVEs -- they do not use paths to generate an explanation. In addition to these qualitative advantages, SVEs generate smaller explanations. In particular, we will show the following:
\begin{enumerate}
    \item Consider cases where $P$ is generated by generating the shortest path, deleting some number of arcs along that path, and repeating. We call such scenarios \emph{closure scenarios}. In closure scenarios, SVEs are always a subset of PBEs (Theorem \ref{thm:k-path}.) Furthermore, experimentally, they are often strictly smaller than PBEs.
    \item When $P$ is instead generated by generating the shortest path, adding a small delay to each segment along that path, and repeating (which we call \emph{incident scenarios}), SVEs are no longer guaranteed to be smaller than PBEs, but we observe experimentally that they are usually much smaller than PBEs.
\end{enumerate}
We run our experiments in multiple geographies and observe similarly positive results.
\subsection{Related Work}

The areas of explainable and interpretable optimization and learning have seen a tremendous interest in recent years, driven by the tremendous growth in algorithmic decision making and the increased complexity of the associated algorithms~\cite{Samek-et-al,molnar2020interpretable,doshi2017towards,hoffman2018metrics,dwivedi-etal}. While there has been interest in changing the algorithms to make them more explainable, e.g., for clustering~\cite{moshkovitz2020explainable}, our focus in this work is on explaining the output of an arbitrary route-selection algorithm.

Our approach to explaining shortest paths is closely related to \emph{shortest-path interdiction problems} and the pioneering work of Fulkerson and Harding~\cite{FH77} 
which directly inspires our approach. In these problems, we are given an interdiction budget $B$ and a cost for interdicting each edge, the goal is to make the resulting shortest path as long as possible, subject to the interdiction cost being at most $B$. Fulkerson and Harding \cite{FH77} consider the model where the edges can be interdicted fractionally (thereby raising their lengths linearly); this allows them to formulate the problem as a linear program, which they solve using a min-cost flow algorithm.  One difference between their setting and ours is that we are given a specific path $P$ which we want to explain, and our goal is to make $P$ a shortest path. More recent investigations of shortest-path interdiction consider the binary setting, where edge can either be retained or deleted (at some \emph{cost} $c(e)$), and the total cost of the deleted edges is at most $B$; these problems are NP-hard in general~\cite{BGV89,KBBEGRZ} and only bicriteria approximations are known for them~\cite{CZ17}. Approaches based on mixed-integer programs have also been considered, e.g., see~\cite{IW02}. 

The work of Forel et al.~\cite{forel2023explainable,forel2024don} uses a counterfactual explanation methodology to explain solutions to data-driven problems, for random forest and nearest-neighbor predictors. (They ask for alternative contexts in which the previous expert-based solution is better than the data-driven solution.) Aigner et al.~\cite{aigner2024framework} study the problem of generating explainable solutions in settings where the same problem is repeatedly solved, by considering two optimization criteria simultaneously: the objective value, and the explainability. The explainability of a solution is measured by comparing the current solution against solutions that were implemented in the past in similar situations.

Erwig and Kumar~\cite{erwig2024explanations} give explanations for combinatorial optimization problems by presenting foil solutions which are plausible alternative solutions, and show that these foils have worse solution value. While the two are similar in spirit, our approach can be considered as a rigorous approach to generating these foil solutions, and also to generating a simultaneously succinct explanation of the proposed path versus all these foils.


\section{Algorithm}\label{sec:algo}

\newcommand{\Residual}{\textsc{Residual}\xspace}
\newcommand{\Modify}{\textsc{Modify}\xspace}
\newcommand{\CutCertificate}{\textsc{CutCert}\xspace}

Consider the problem where we are given a directed graph $G$, weights $\ell, u$ on the arcs, a path $P$ from $s$ to $t$, and some weighting of the arcs $\tau$. Without loss of generality, we may assume that $\tau(e) = 0$ for all non-pliable arcs, as $w_e = \ell(e)$ for those arcs anyways. We want to find a valid explanation $w$ for path $P$ with minimum valuation. For this, recall the linear program \ref{LP1} from the introduction with variables $w_e$ and $d_v$: (We call this the \emph{cut formulation} of the problem.)
\begin{align}
    \min \sum_{e \in E(G)} & \tau(e)(w_e - \ell(e)) && \tag{LP1} \label{LP1}\\
    w_e &\le u(e) && \forall e \in E(G) \label{eq:1}\\
    w_e &\ge \ell(e) && \forall e\in E(G) \label{eq:2} \\
    d_v - d_u &\leq w_e & &\forall e = (u,v) \in E(G) \label{eq:3} \\
    d_v - d_u &= w_e && \forall e = (u,v) \in P. \label{eq:4}
\end{align}

\begin{theorem}
    The optimal solution for the linear program \ref{LP1} is a $\tau$-SVE for the path $P$.
\end{theorem}

\begin{proof}
    The objective function of the linear program is the valuation of the path with respect to $w$, and hence the minimization ensures the simplicity property. Constraints \eqref{eq:1} and \eqref{eq:2} enforce validity of the weight function $w$, so it remains to show the sufficiency property, i.e., that $P$ is indeed a shortest $s$-$t$ path with respect to the weight function $w$. Indeed, consider any path $Q$, and let $e_1, e_2, \ldots, e_k$ be the arcs on $Q$, where each $e_i = (u_{i-1}, u_{i})$. Then $w(Q) = \sum_{j=1}^k w_{e_j} \geq \sum_{j=1}^k d_{u_{i}} - d_{u_{i-1}}$, using constraint \eqref{eq:3}. Using that $u_{k} = t$ and $u_0 = s$, the resulting telescoping sum gives $w(Q) \geq d_t - d_s$. A similar calculation for the path $P$, now using constraint \eqref{eq:4} gives $w(P) = d_t - d_s \leq w(Q)$, which shows that $P$ is a shortest path with respect to $w$, and completes the proof. 
\end{proof}

\subsection{Runtime of Solving the Cut Formulation}

While the cut formulation can be solved using an off-the-shelf LP solver like CPLEX or Gurobi for small instances, road networks contain millions of nodes and arcs, leading to cut formulation instances with millions of constraints. CPLEX and Gurobi are not suitable for such large instances. Instead, we use the special structure of the cut formulation to get a fast combinatorial solver that can be applied to graphs with millions of nodes.

\subsection{Choosing $\tau$}

To complete the specification of the cut formulation, we need to pick $\tau$. In this section, we discuss the considerations behind that choice. For simplicity, we require that the entries of $\tau$ be integers. This can be achieved by rounding. The runtime of our algorithm empirically is linear in $\tau_{\max} := \max_{e\in E(G)} \tau(e)$. This leads to a few candidates for $\tau(e)$ for all $e\in E(G)$:

\medskip
\noindent\textbf{Option 1:} $\tau(e) = 1$. This leads to a fast algorithm, as $\tau_{\max} = 1$. Intuitively, a good explanation should consist of a small number of arcs, which does not perfectly align with the $\tau$-valuation objective. While this valuation works well in simple scenarios, it does not always lead to short explanations, so we do not investigate this value of $\tau$ in this work.

\medskip
\noindent\textbf{Option 2:} $\tau(e) = \frac{1}{u(e) - \ell(e)}$. This setting has the useful property that for any $e$ and any valid explanation $w$, $0\le V_{\tau}(w,e) \le 1$. As a result, for any valuation $w$ with $w_e = u(e)$ or $w_e = \ell(e)$ for all $e$, $V_{\tau}(w)$ is exactly the number of arcs for which $w_e \ne \ell(e)$; i.e. the size of the explanation. Thus, the valuation of an explanation relates to its size, at least for certain $w$s. 
However this choice has one main downside: rounding $\tau$s requires setting a high value for $\tau_{\max}$, which significantly affects the running time of our algorithms. 
Nevertheless, our theoretical results all work for this choice of $\tau$.
\medskip
\noindent\textbf{Option 3:} $\tau(e) = 1 + \left\lfloor \frac{C_0 \ell(e)}{u(e)}\right\rfloor$ for some constant $C_0 \ge 1$. When $C_0 = 1$, all pliable arcs have the same weight, so this is the same as Option 1. However, for $C_0$ approaching infinity, this resembles Option 2, both in runtime and quality, with this value chosen instead of $\tau(e) = 1 + \lfloor \frac{C_0}{u(e) - \ell(e)}\rfloor$ to be scale-invariant. On the other hand, even for segments with very large $u(e)$, this choice still applies some penalty. While we cannot show theoretical results about this choice, we do use this choice in our experiments due to the high runtime of Option 2. To trade off between runtime and quality, we use $C_0 = 10$. In spite of this choice being somewhat close to 1, we experimentally show that $\tau$-SVEs are small anyways.\\

Our combinatorial algorithm is based on the classical augmenting paths framework for solving flow problems. We give our combinatorial algorithm for solving the cut formulation in the appendix of the full version of this paper.

\subsection{Subgraphs}

For the experiments in this paper, we ran our combinatorial algorithm on the full road network. To achieve an order-of-magnitude runtime improvement and enable real-time use while generating similar explanations, one can instead run our algorithm on a subgraph consisting of all segments that are likely to be relevant to the given query pair. On one choice of subgraph, our combinatorial algorithm takes 50 milliseconds for each query pair in our Seattle query set when run on an Apple M1 Pro. Thus, our algorithm is suitable for real-time use.

\section{Closures v.s. Incidents and SVEs v.s. PBEs}\label{sec:closures}

In this section, we give theoretical results comparing PBEs to SVEs that show that SVEs always perform at least as well as PBEs in both single-closure and multi-closure settings. In Sections \ref{sec:examples} and \ref{sec:experiments}, we will see experimentally that SVEs do well at explaining routes that were diverted due to both closures and delays. Our theoretical results primarily focus on closures. In both our theoretical and experimental analysis, we answer one key question: how small are SVEs? We use PBEs as a benchmark to assess relative size.

\subsection{Our Evaluation Setup}\label{subsec:eval-setup}

A \emph{scenario} is a tuple $(G,\ell,u,P)$ for which we would like a valid explanation of the path $P$. We construct two different types of scenarios on which to evaluate our algorithm: \emph{closure scenarios} and \emph{incident scenarios}. We define each of these now:\\

\textbf{Closure scenarios}. Let $G$ be an arbitrary (no-traffic) $\ell$-weighted digraph. Pick a positive integer $k$, representing the number of closures generated. Pick $s,t\in V(G)$. Initialize $y_e = \ell(e)$ and let $P_0$ be the $y$-weighted shortest $s$-$t$ path in $G$. For $i=1,2,\hdots,k$, let $e_i$ be an arbitrary edge on $P_{i-1}$, delete $e_i$ (i.e. set $y_{e_i} = \infty$), let $P_i$ be the shortest $y$-weighted $s$-$t$ path in $G$, and increment $i$. Then, let $u(e)$ be an arbitrary number\footnote{$u(e)$ is sometimes set higher in our experiments in order to introduce additional irrelevant traffic to make the explainer's task harder.} for which $u(e) \ge y_e$ for all arcs $e\in E(G)$ and let $P := P_k$. $u(e)$ is required to be finite for all $e\ne e_i$ for some $i$ and $u(e_i) = \infty$ for all $i$.

Note that closure scenarios have some flexibility induced by the choice of $e_i$ for each $i$. In our experiments, we delete multiple segments around $e_i$ -- not just $e_i$ -- before generating $P_i$ in order to increase the diversity of the $P_i$s. In our theoretical results, we only consider deletion of a single $e_i$ to obtain $P_i$. This is critical to Theorem \ref{thm:two-path}, but not important to Theorem \ref{thm:k-path}.\\

\textbf{Incident scenarios}. Let $G$ be an arbitrary $\ell$-weighted digraph. Pick a positive integer $k$, $\gamma > 1$, and $s,t\in V(G)$. Initialize $y_e = \ell(e)$ and $P_0$ to be the shortest $y$-weighted $s$-$t$ path in $G$. For $i=1,2,\hdots,k$, reset $y_e = \gamma y_e$ for all $e\in P_{i-1}$ and let $P_i$ be the shortest $y$-weighted $s$-$t$ path in $G$. Let $P := P_k$ and $u(e)$ be an arbitrary number for which $u(e) \ge y_e$ for all arcs $e\in E(G)$.

Note that incident scenarios uniquely specify the list of paths given $\gamma,s,t$, and $k$, unlike with closure scenarios. We further specify $u, \gamma$, and the $e_i$s in our experiments, but prove results in this section given arbitrary choices for these quantities.

\subsection{Theoretical Results: Single Closure}

We now show that, for all three options for $\tau$ presented in Section~\ref{sec:algo}, SVEs are at least as small as PBEs in the two-path setting, where there is a single closure.

\begin{theorem}\label{thm:two-path}
Let $(G,\ell,u,P = P_1)$ be a closure scenario with $k=1$. Suppose that $\tau$ satisfies both
\begin{enumerate}[nosep]
    \item $\tau(e) > 0$ for all $e\notin P_0$
    \item $\tau(e) \ge \tau(e_1)$ for all $e\in P_0$.
\end{enumerate}
Let $x$ and $w$ be the PBE and a $\tau$-SVE for this scenario respectively. Then, $\text{support}(w) \subseteq \text{support}(x)$. In fact, $\text{support}(w)$ entirely consists of arcs $e$ of $P_0$ for which $\tau(e) = \tau(e_1)$.
\end{theorem}

Before we prove this, let us consider some possible settings for $\tau$. Setting $\tau(e) = 1$ for all $e$ satisfies both criteria trivially. Setting $\tau(e) = \frac{1}{u(e) - \ell(e)}$ also satisfies both criteria, as $\tau(e_1) = 0$ \footnote{Recall that $u(e_1) = \infty$.} and $\tau(e) > 0$ for all $e\ne e_1$. Finally, setting $\tau(e) = 1 + \left\lfloor \frac{C_0 \ell(e)}{u(e)}\right\rfloor$ also satisfies the criteria, as $\tau(e_1) = 1$ and $\tau(e) \ge 1$ for all $e$. Thus, all three options for $\tau$ considered in this paper satisfy the assumptions in the theorem.

\begin{proof}
First, notice that $\text{support}(x)$ is precisely the set of arcs $e$ in $P_0 \setminus P_1$ for which $u(e) \ne \ell(e)$. By the simplicity criterion, $\text{support}(w) \cap P_1 = \emptyset$. Thus, it suffices to show that $\text{support}(w) \subseteq P_0$; i.e. that $w_f = \ell(f)$ for all $f\notin P_0$.

Suppose, for the sake of contradiction, that there exists an arc $f\notin P_0$ for which $w_f > \ell(f)$. By definition of $\tau$, $\tau(f) > 0$. By the validity and sufficiency properties of $w$ respectively, $\ell(P_1) \le w(P_1) \le w(P_0)$. Thus,
\begin{align*}
    V(w) &= \sum_{e\in E(G)} \tau(e) (w_e - \ell(e))\\
    &\ge \tau(f) (w_f - \ell(f)) + \tau(e_1) (w(P_0) - \ell(P_0))\\
    &> \tau(e_1) (\ell(P_1) - \ell(P_0))
\end{align*}

We now show that there is a different valid explanation $z$ for $P_1$ with the property that $V(z) = \tau(e_1)(\ell(P_1) - \ell(P_0))$. Specifically, let $z_{e_1} := \ell(e_1) + \ell(P_1) - \ell(P_0)$ and $z_e = \ell(e)$ for all $e\ne e_1$. This satisfies the validity constraint, so it suffices to check sufficiency. Consider an $s$-$t$ path $K$. We check that its length according to $z$ is longer than $P_0$'s by considering two cases:

\begin{enumerate}
    \item $e_1\in K$: If $e_1\in K$, note that by validity of $z$ and the optimality of $P_0$,
\begin{gather*}
    z(K) = z_{e_1} + z(K\setminus \{e_1\}) \ge z_{e_1} - \ell(e_1) + \ell(K) \\ \ge z_{e_1} - \ell(e_1) + \ell(P_0) = \ell(P_1) = z(P_1)
\end{gather*}    
    
    \item $e_1\notin K$: If $e_1\notin K$, note that $z(K) = \ell(K) \ge \ell(P_1) = z(P_1)$, where the inequality follows from the fact that $P_1$ is the shortest path in the graph with $e_1$ deleted.
\end{enumerate}

Thus, in both cases, $z(K) \ge z(P_1)$, meaning that $P_1$ is a shortest path according to $z$, implying the sufficiency of $z$. Thus, $z$ is a valid explanation for $P_1$. Furthermore, $V(z) = \tau(e_1)(\ell(P_1) - \ell(P_0)) < V(w)$, contradicting the simplicity of $w$. Thus, $f$ cannot exist, as desired.
\end{proof}

The fact that $w$ is supported on arcs of $P_0$ for which $\tau(e) = \tau(e_1)$ highlights that small changes in $\tau$ can result in large decreases in support size---a fact that Option 3 for $\tau$ takes advantage of.

\subsection{Theoretical Results: Multiple Closures}

We now prove that if $\tau$ is chosen in a way that satisfies Option 2, then SVEs are a subset of PBEs even under multiple closures. 

\begin{theorem}\label{thm:k-path}
Pick an arbitrary positive integer $k$ and let $(G,\ell,u,P = P_k)$ be a closure scenario in which all arcs in $P_0\cup\hdots\cup P_k\setminus\{e_1,e_2,\hdots,e_k\}$ are not pliable. Fix a $\tau$ with the following properties:
\begin{enumerate}[nosep]
    \item $\tau(e_i) = 0$ for all $i\in \{1,2,\hdots,k\}$
    \item $\tau(e) > 0$ for all pliable $e\notin \{e_1,e_2,\hdots,e_k\}$.
\end{enumerate}
Let $x$ and $w$ be the the PBE and a $\tau$-SVE for this scenario respectively. Then, $\text{support}(w) \subseteq \text{support}(x)$.
\end{theorem}

The non-pliability assumption is only there to constrain the behavior of the PBE algorithm so that $\text{support}(x) = \{e_1,\hdots,e_k\}$. This underscores the brittleness of the PBE algorithm. Furthermore, the Option 2 choice of $\tau$ satisfies Properties (1) and (2) due to the $e_i$s having infinite $u(e)$ and all other arcs having finite $u(e)$. The other options for $\tau$ do not satisfy Properties (1) and (2) of Theorem \ref{thm:k-path}.

\begin{proof}
The PBE algorithm in iteration $i$ increments $x_e$ only for $e = e_i$. Thus, the $Q$s found in the PBE algorithm are exactly the $P_i$s, which means that $\text{support}(x) = \{e_1,\hdots,e_k\}$. Therefore, it suffices to show that $\text{support}(w) \subseteq \{e_1,\hdots,e_k\}$. Suppose, for the sake of contradiction, that there exists $f\in \text{support}(w) \setminus \{e_1,\hdots,e_k\}$. Then, $V(w) \ge \tau(f) (w_f - \ell(f)) > 0$ by the second condition on $\tau$ and the definition of $\text{support}(w)$. However, $V(x) = 0$, contradicting the optimality of $w$. Thus, $\text{support}(w) \subseteq \{e_1,\hdots,e_k\}$ as desired.
\end{proof}

\section{Examples}\label{sec:examples}

We now give examples that qualitatively demonstrate that the algorithm from Section \ref{sec:algo} produces compelling explanations. We demonstrate this in both simple and complex scenarios.

\subsection{Simple Examples: Single Road Closure}

Theorem \ref{thm:two-path} implies that in the presence of one blockage, SVEs do not pick irrelevant blockages to be a part of the explanation, as we intuitively expect. We now illustrate this fact using one scenario from Seattle and another from the German state of Baden-W{\"u}rttemburg, with $\tau(e) = 1 + \lfloor \frac{C_0 \ell(e)}{u(e)}\rfloor$ for all $e$ (Option 3). In both examples, we pick one arc along with some neighboring arcs along the path to delete. Furthermore, for all $e\notin P_0\cup P_1$, $u(e) = 2\ell(e)$. Thus, the explanation could include segments outside of $P_0$, but we observe that it does not, in accordance with Theorem \ref{thm:two-path}. In the first example, a bridge is blocked and the SVE correctly includes it in the explanation. In the second example, a surface street is blocked and the SVE correctly includes it. The PBE in both examples contains the found SVE, as guaranteed by the theorem. This is in spite of the fact that Theorem \ref{thm:two-path} does not exactly apply, as we delete more than one arc to obtain the path $P$.  Figure~\ref{fig:single-closure} illustrates this observation.

\begin{figure}[H]
    \centering
    \begin{subfigure}{.58\columnwidth}
    \includegraphics[width=0.95\columnwidth]{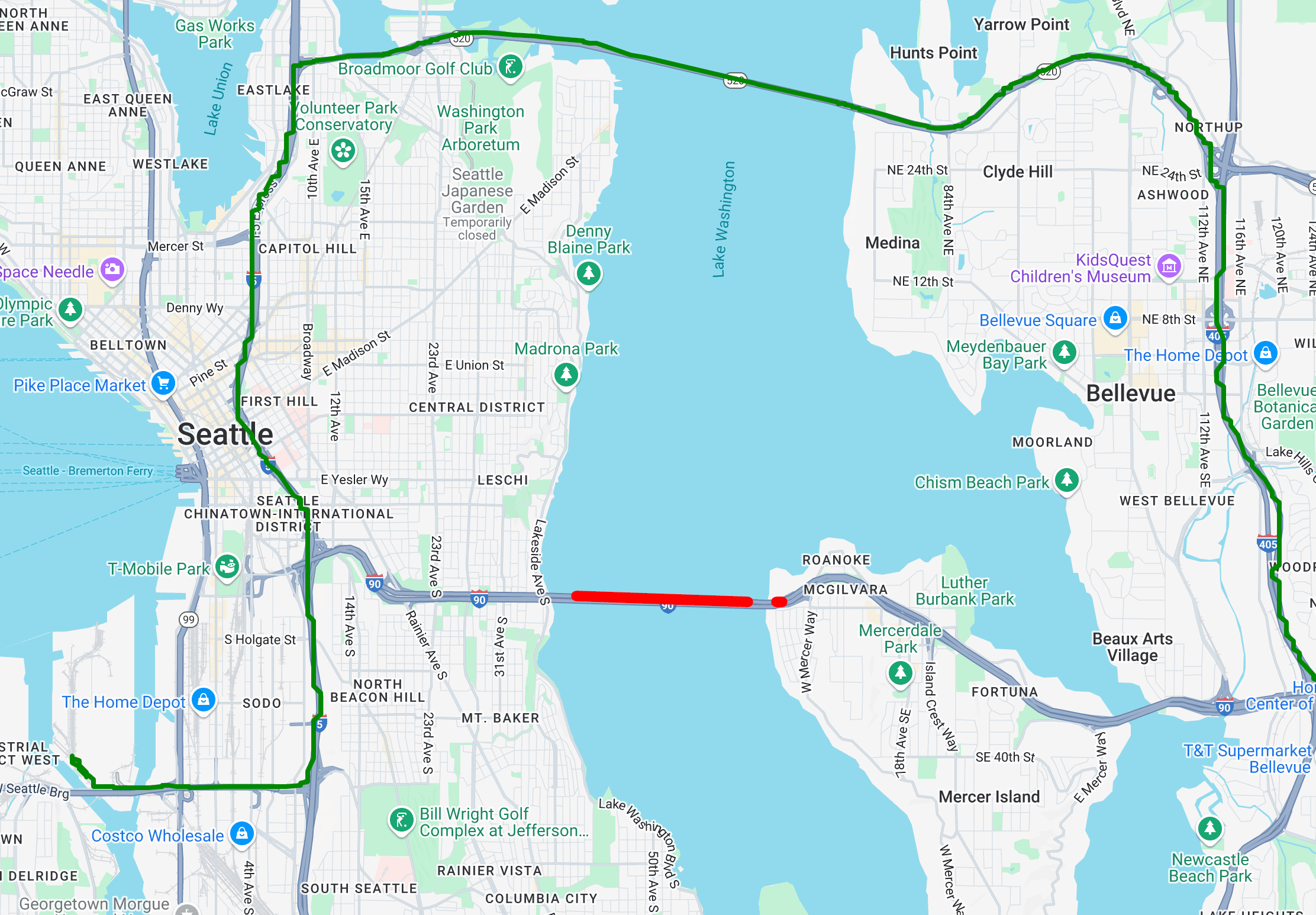}
    \caption{Seattle}
    \label{fig:single-closure-1}
    \end{subfigure}%
    \begin{subfigure}{.42\columnwidth}
    \includegraphics[width=0.95\columnwidth]{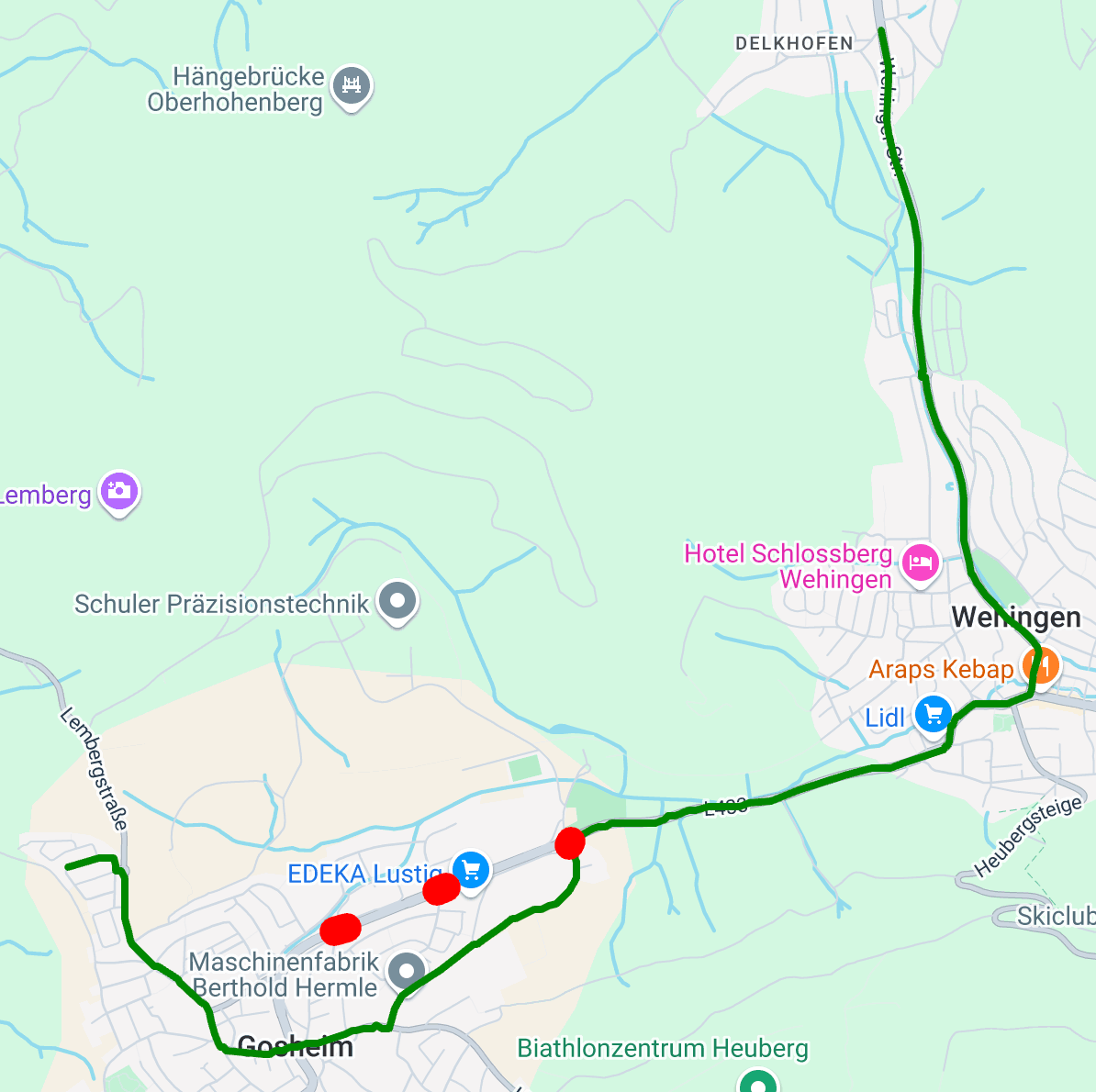}
    \caption{Baden}
    \label{fig:single-closure-2}
    \end{subfigure}
    \caption{Single-closure scenarios. The green route is the shortest traffic-aware path, while the red segments are the SVE. Note that the SVE is contained in the shortest no-traffic path in both settings, yielding the conclusion of Theorem \ref{thm:two-path} despite the conditions not quite being satisfied. In particular, the SVE is more succinct than the PBE.
    }
    \label{fig:single-closure}
\end{figure}

\subsection{Multiple Incidents}

In the previous section, we learned that SVEs do not contain irrelevant segments when one road segment is closed. This is no longer true when multiple road segments are closed: SVEs may not be confined to the incident set---for an example, see \Cref{fig:two-graphs}, as Theorem \ref{thm:k-path} does not apply for the Option 3 $\tau$. Nonetheless, we can still assess the size of the closure set. Figure~\ref{fig:multi-path} illustrates that SVEs typically contain a small fraction of the number of incidents present in the network. Figure~\ref{fig:api} illustrates that SVEs only contain traffic that is relevant to the path $P$.

\begin{figure}[H]
    \centering
    \begin{subfigure}{.53\columnwidth}
    \includegraphics[width=.95\columnwidth]{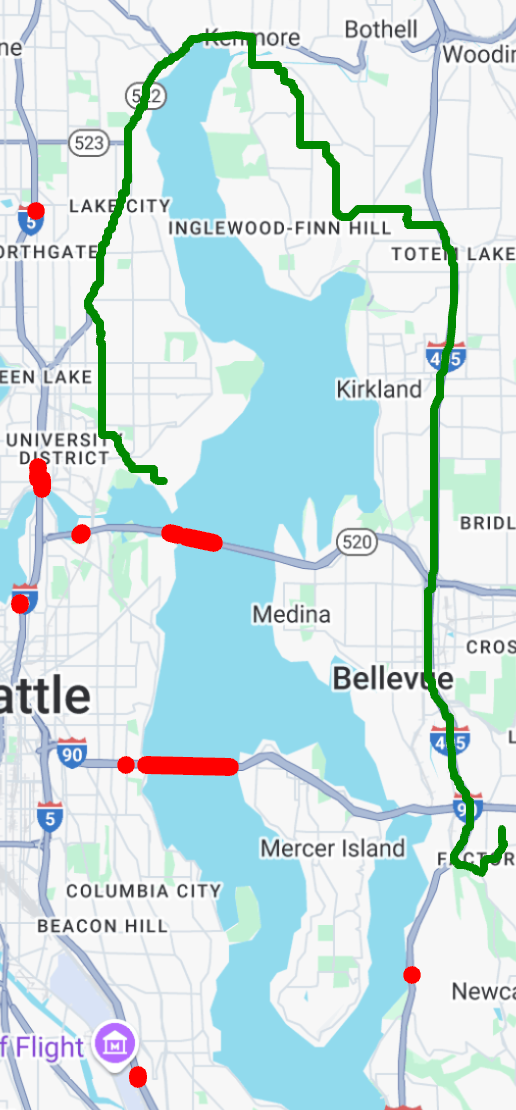}
    \caption{Seattle}
    \end{subfigure}%
    \begin{subfigure}{.44\columnwidth}
    \includegraphics[width=.95\columnwidth]{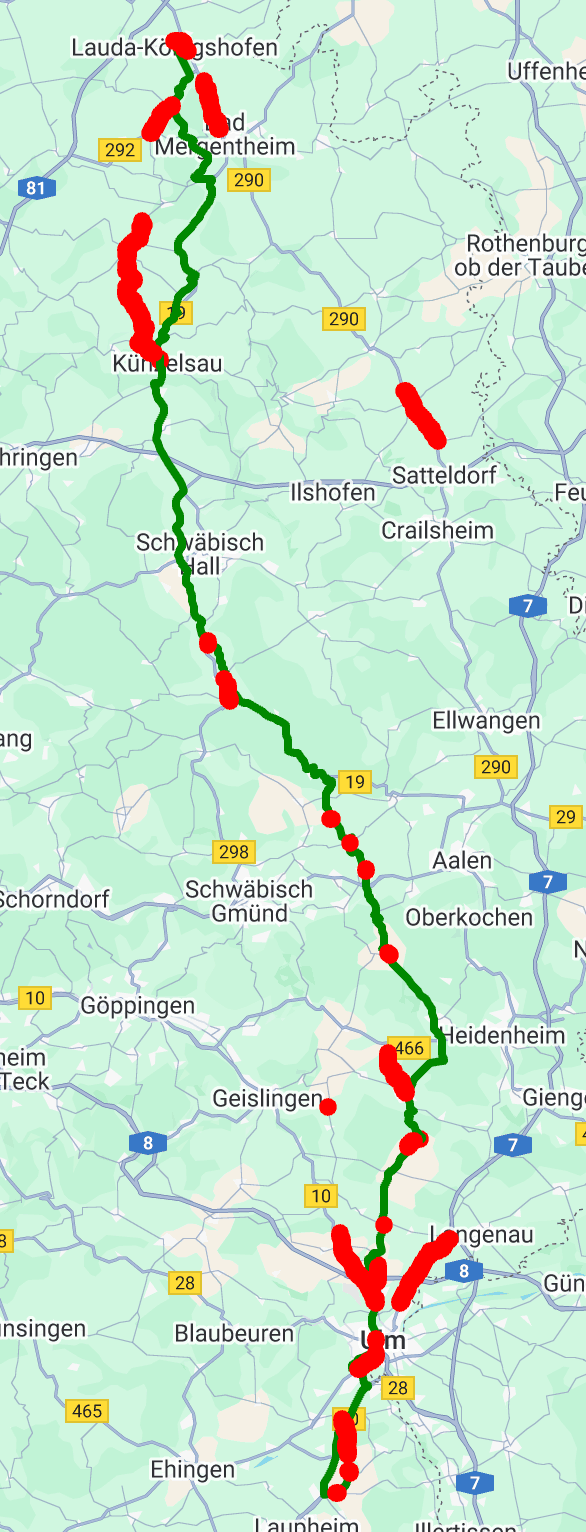}
    \caption{Baden}
    \end{subfigure}
    \caption{$k = 9$ multi-closure and incident scenarios for Seattle and Baden Long respectively. The explanation perfectly coincides with closures in Seattle. In Baden, we obtain a similar result even for incidents.}
    \label{fig:multi-path}
\end{figure}

\begin{figure}[H]
    \centering
    \begin{subfigure}{.45\columnwidth}
    \includegraphics[width=0.95\columnwidth]{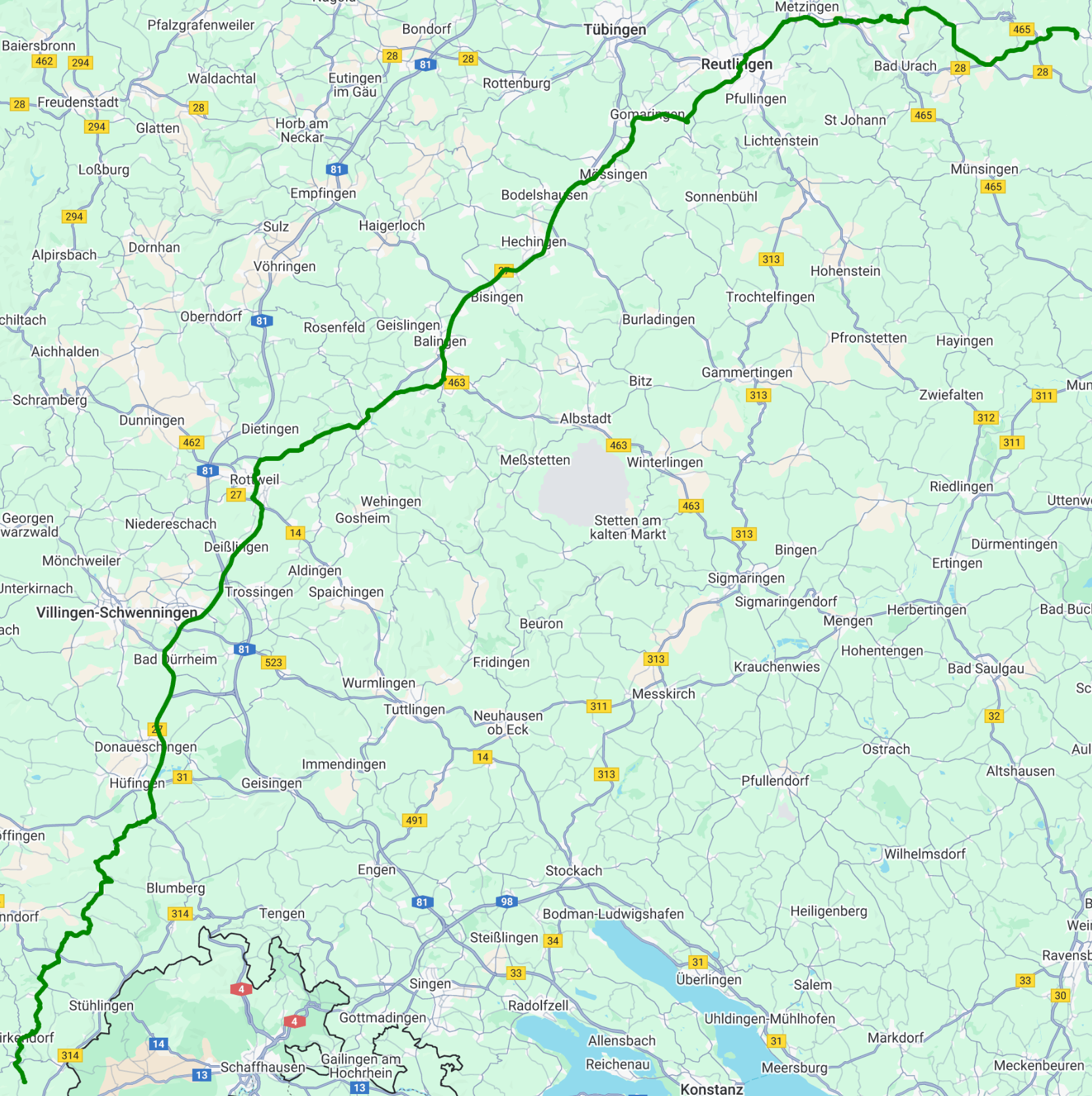}
    \caption{Input to algorithm}
    \end{subfigure}
    \begin{subfigure}{.45\columnwidth}
    \includegraphics[width=0.95\columnwidth]{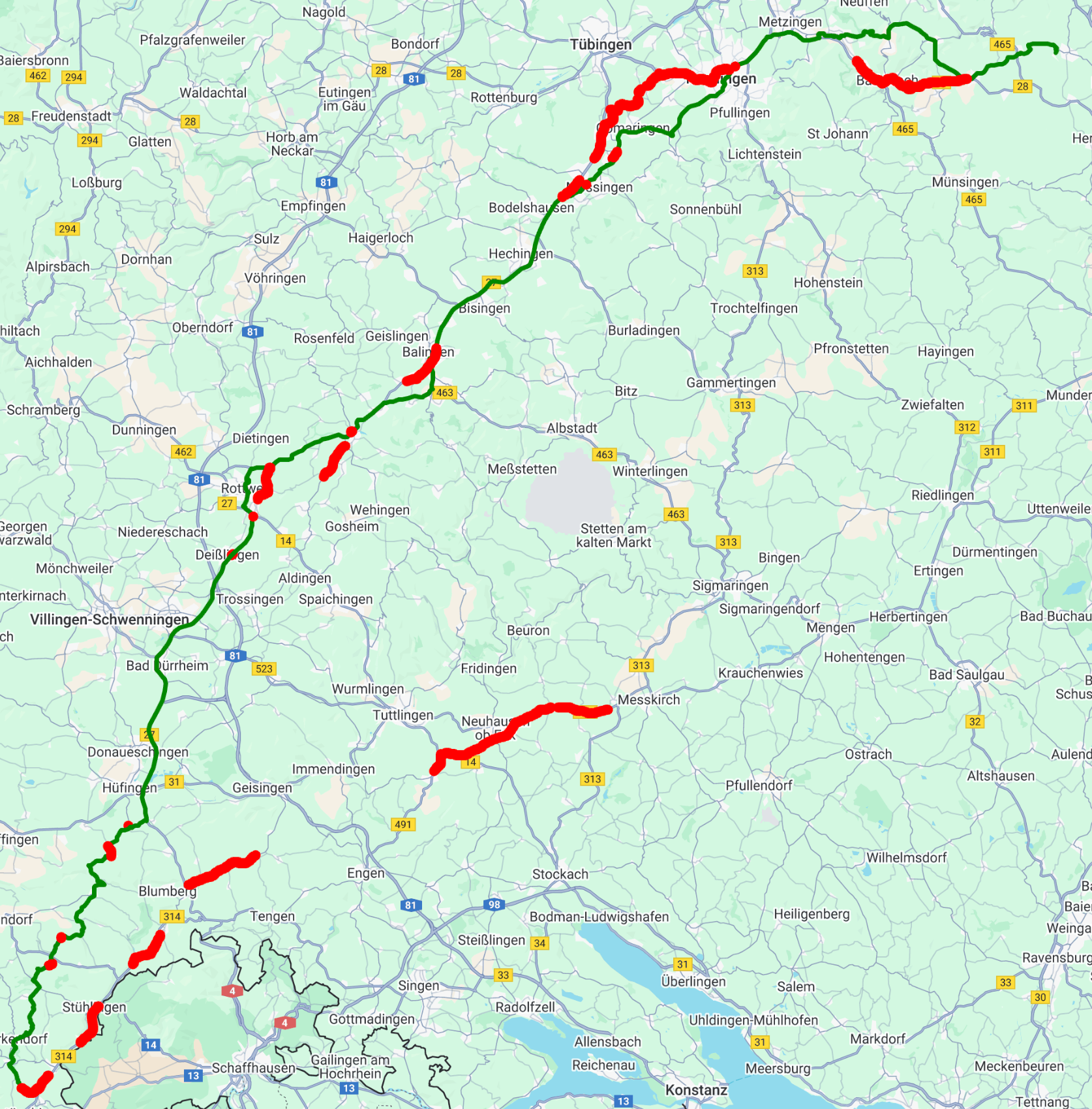}
    \caption{Output of algorithm}
    \end{subfigure}
    \begin{subfigure}{.45\columnwidth}
    \includegraphics[width=0.95\columnwidth]{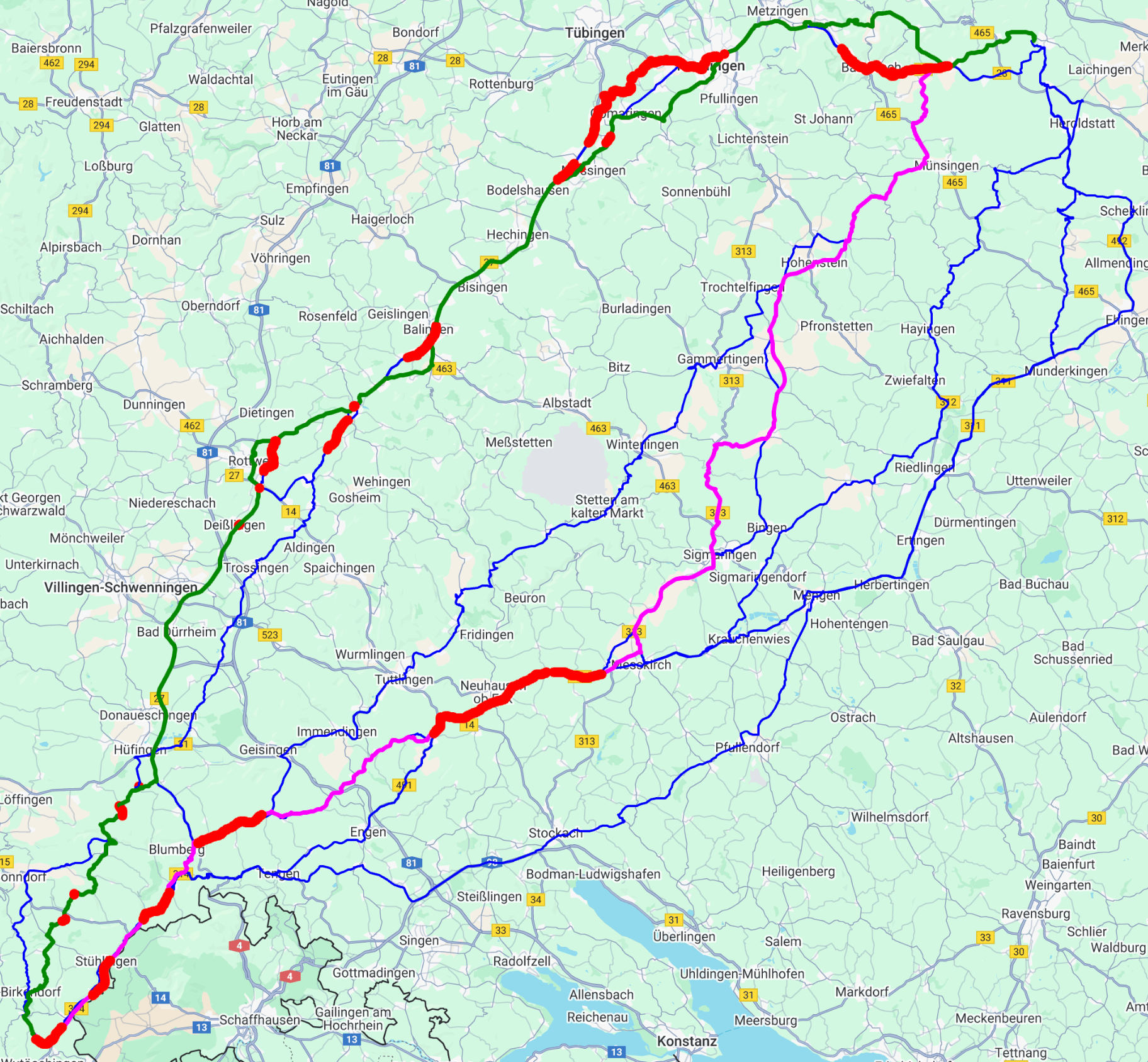}
    \caption{Scenario context}
    \end{subfigure}
    \caption{The explainer algorithm is just given the green path $P$ and outputs a list of red segments that explain the optimality of the green path. All of the explanation segments are relevant, as they all lie on paths used to construct the scenario, especially the shortest no-traffic path (pink path).}
    \label{fig:api}
\end{figure}

\section{Experiments}\label{sec:experiments}

In Section \ref{sec:examples}, we saw that SVEs yield compelling explanations in a variety of examples. In this section, we quantitatively study this phenomenon. We evaluate SVEs on datasets consisting of three different types of scenarios:
\begin{enumerate}
    \item 1 and 9-closure scenarios ($k=1,9$) in which closure arcs and non-path arcs are pliable.
    \item 1 and 9-closure scenarios ($k=1,9$) in which all arcs are pliable.
    \item 10-path ($k=9$) incident scenarios in which all arcs are pliable.
\end{enumerate}
For the first category, PBEs often uncover optimal or near-optimal explanations, so we focus on illustrating that SVEs do just as well in accordance with the conclusions of Theorems \ref{thm:two-path} and \ref{thm:k-path}, even though they do not apply for our chosen $\tau$ (Option 3). For the second category, we show that SVEs typically perform much better than PBEs, as the added flexibility only confounds the PBE algorithm. Finally, for the last category, we also show that SVEs produce a much smaller explanation than PBEs. In understanding the behavior of SVEs, sufficiency guarantees that an SVE is large enough to explain the choice of the desired path, so the question that all of our experiments answer is the following: how small are SVEs?

\subsection{Experiment Setup}

We use a road network derived in Open Street Maps (OSM). In this graph, each vertex represents a road segment and each arc represents a valid transition (e.g. turn) from one road segment to another. We consider two different regions:
\begin{enumerate}
    \item Baden-W{\"u}rttemburg, a state in southwestern Germany. Baden has 6.4 million nodes and 13.5 million arcs.
    \item Washington State, USA. Washington has 7.4 million nodes and 15.4 million arcs.
\end{enumerate}
In both regions, we use the following three sets of 100 queries:
\begin{enumerate}
    \item A short set, where the origin-destination (OD) pairs were drawn uniformly subject to the constraint that the crow's flight distance between them is between 1 and 3 miles.
    \item A medium set, where OD pairs were drawn randomly subject to being between 5 and 20 miles apart.
    \item A long set, where OD pairs were drawn randomly subject to being between 80 and 120 miles apart.
\end{enumerate}
In Washington only, we generated an additional set of 100 randomly drawn queries that are between 5 and 30 miles apart within the Seattle metro area.

%

Baden-W{\"u}ttemburg and Washington have diameter roughly 180 and 400 miles respectively as the crow flies. While Baden does not have many natural obstructions, Washington has many. In particular, some of the query pairs in the medium set actually have much longer routes than 20 miles due to the presence of Puget Sound. Our Seattle dataset also interacts with natural obstructions like Lake Washington, but in not as severe of a way.

For each of these query sets and each OD pair in each query set, we generate the three scenarios discussed earlier, obtaining three different experiments. In each of those experiments, we collect different metrics, so we continue our setup discussion in those sections.

\subsection{Road Closure Scenarios with Few Pliable Arcs}

\noindent 
{\bf Setup.}
For each query pair set $\mc Q$ with associated road network digraph $G$ (Baden or Washington) with (no-traffic) arc weights $\ell$ and each $(s,t)\in \mc Q$, we construct a closure scenario for $k=1$ or $k=9$ (2 or 10 paths) as defined in Section \ref{subsec:eval-setup}. In order to fully specify the algorithm, we need to choose $e_i$ for each $i\in \{1,2,\hdots,k\}$. Intuitively, $e_i$ should be the most important arc on $P_{i-1}$. Let $\ell_0(e) := \ell(e)$ for all $e$. We use the following heuristic:

\begin{enumerate}
    \item Let $L = |P_{i-1}|$ and let $W$ be the set of arcs in $P_{i-1}$ that are at least $f(L)$ hops away from the origin and destination of $P_{i-1}$ for some nondecreasing function $f$. This filtering reduces the number of OD pairs for which closed arcs form a cut. 
    \item We assign a road type number to each arc, derived from the ``highway'' field in OSM \cite{osmhighway}. Lower numbers indicate more arterial roads. For instance, 0 indicates a motorway and 1 indicates a motorway entrance or exit.
    \item Let $X$ denote the arcs in $W$ with minimum road type.
    \item Let $e_i$ be the arc in $X$ with maximum length, with ties broken by number of lanes and arbitrarily if necessary.
    \item Let $\ell_i$ be a function defined as follows: for all $e$ within 5 hops of $e_i$ along $P_{i-1}$, let $\ell_i(e) := 10000\ell_{i-1}(e)$. We call all such arcs $e$ \emph{closed arcs}. Let $C_i$ denote the set of these arcs. Otherwise, let $\ell_i(e) := \ell_{i-1}(e)$.
    \item Let $P_i$ denote the shortest $\ell_i$-weighted $s$-$t$ path in $G$.
\end{enumerate}

Label the query pair $(s,t)$ \emph{valid} if and only if $t$ is reachable from $s$ in $G$ and the $C_i$'s for $i\in \{1,2,\hdots,k+1\}$ are disjoint.\footnote{$C_{k+1}$ is defined using one additional iteration of the above loop just for technical reasons -- to delineate why $P\ne P_i$ for any $i\in \{0,1,\hdots,k-1\}$.} This algorithm often chooses $e_i$'s that are bridges for example, as bridges tend to be long highway segments. To complete the specification of the scenario, we need to define $u$ and $P$. Let $u(e) := \ell_k(e)$ for all $e\in P_0\cup \hdots\cup P_k$ and let $u(e) = 2\ell(e)$ for all other $e\in E(G)$ and let $P := P_k$. Note also that $P\ne P_i$ for any $i < k$, as otherwise $C_{k+1} = C_{i+1}$, which would mean that this query pair is invalid. The factor of 2 buffer in the definition of $u(e)$ for all $e\notin P_0\cup \hdots\cup P_k$ makes it so that the SVE could contain arcs $e$ outside of the $P_i$s, though Theorem~\ref{thm:k-path} prevents this for certain choices of $\tau$. We answer the following questions:

\begin{enumerate}
    \item For what fraction of $\mc Q$ is it the case that $\text{support}(w) \subseteq C$, where $w$ is the SVE and $C := C_1\cup C_2\cup\hdots\cup C_k$ is the set of all closed segments?
    \item For what fraction of $\mc Q$ is it the case that $\text{support}(x)\subseteq C$, where $x$ is the PBE?
\end{enumerate}

Theorems \ref{thm:two-path} and \ref{thm:k-path} would show that this fraction is 1 for SVEs, but the conditions of these theorems are not quite satisfied, due to a noncompliant choice of $\tau$. Thus, our experiments assess deviation from the results of these theorems.\\
%

\noindent {\bf Results.}
Note that the PBE algorithm exactly generates the paths $P_0,P_1,\hdots,P_{k-1},P_k$, so $x(e) = u(e)$ for all $e\in C_1\cup\hdots\cup C_k$ and $x(e) = \ell(e)$ otherwise. In our experiment, we also see that $w(e) = u(e)$ only for $e\in C_1\cup\hdots\cup C_k$ for all query pairs in $\mc Q$, so $\text{support}(w)\subseteq \text{support}(x)$ in all cases. The percent of valid pairs is listed and is always at least 77\%. Our results yield the conclusion of Theorem \ref{thm:k-path} in spite of the fact that the conditions are not satisfied.  We refer the reader to Table~\ref{tab:result1}.

\begin{table}[H]
\centering
\begin{tabular}{|l|c|c|c|c|c|}
\hline
\% (\textbf{exp} $\subseteq$ \textbf{C})  & \% \textbf{val} & \multicolumn{2}{c|}{\textbf{SVE}} & \multicolumn{2}{c|}{\textbf{PBE}}\\ \hline
\# paths & $\le 10$ & 2 & 10 & 2 & 10 \\ \hline
Seattle & 98\% & 100\% & 100\% & 100\% & 100\% \\ \hline
WA Short & 75\% & 100\% & 100\% & 100\% & 100\% \\ \hline
WA Med & 86\% & 100\% & 100\% & 100\% & 100\% \\ \hline
WA Long & 86\% & 100\% & 100\% & 100\% & 100\% \\ \hline
Baden Short & 77\% & 100\% & 100\% & 100\% & 100\% \\ \hline
Baden Med & 98\% & 100\% & 100\% & 100\% & 100\% \\ \hline
Baden Long & 98\% & 100\% & 100\% & 100\% & 100\% \\ \hline
\end{tabular}

\caption{
Percentage of valid query pairs for which the our SVE-based results and the baseline PBE are contained in the closed set. In this table and in Table \ref{tab:result2}, $C$ denotes the closed set in each example. We note that the SVE aligns with the closed set 100\% of the time and is thus always as small as the PBE. It is not typically strictly smaller, though, as the PBE is often nearly optimal.
}
\label{tab:result1}
\end{table}

\subsection{Road Closure Scenarios with Many Pliable Arcs}

\noindent 
{\bf Setup.} 
We repeat the experiment from the previous section with exactly one crucial modification. Instead of setting $u(e) = 2\ell(e)$ just for $e\notin P_0\cup P_1\cup \hdots\cup P_k$, we set $u(e) := 2\ell(e)$ for all $e$ with the property that $\ell_k(e) = \ell(e)$. This makes it so that \emph{every} arc in $G$ is pliable, whereas the previous setting made it so that only $e_i$s and nearby arcs along the $P_i$s were pliable (besides all non-path arcs). \\

\noindent {\bf Results.} 
We compute the same metric as before: percentage of valid query pairs for which the explanation is contained in the closed set. SVEs perform significantly better than PBEs. When $|P_0| \ge 12$ \footnote{12 here stems from the 5 hops choice earlier.} -- which is always the case in our datasets -- PBEs contain all pliable arcs on $P_0$. Thus, PBEs always contain arcs outside of the closed set leading to larger explanations. On the other hand, SVEs almost always just contain closed arcs (see Table ~\ref{tab:result2} and Figure~\ref{fig:worst-baden-medium}).


\begin{table}[H]
\centering
\begin{tabular}{|l|c|c|c|c|c|}
\hline
\% (\textbf{exp} $\subseteq$ \textbf{C})  & \% \textbf{val} & \multicolumn{2}{|c|}{\textbf{SVE}} & \multicolumn{2}{|c|}{\textbf{PBE}}\\ \hline
 \# paths & $\le 10$ & 2 & 10 & 2 & 10 \\ \hline
Seattle & 98\% & 100\% & 100\% & 0\% & 0\% \\ \hline
WA Short & 75\% & 100\% & 96.0\% & 0\% & 0\% \\ \hline
WA Med & 86\% & 100\% & 100\% & 0\% & 0\% \\ \hline
WA Long & 86\% & 100\% & 100\% & 0\% & 0\% \\ \hline
Baden Short & 77\% & 100\% & 93.5\% & 0\% & 0\% \\ \hline
Baden Med & 98\% & 100\% & 98.0\% & 0\% & 0\% \\ \hline
Baden Long & 98\% & 100\% & 100\% & 0\% & 0\% \\ \hline
\end{tabular}

\vspace{0.05in}
\caption{
Higher $u(e)$ on some arcs. This additional flexibility makes SVEs smaller than PBEs. 
}
\label{tab:result2}
\end{table}

\begin{figure}[H]
    \centering
    \includegraphics[width=0.5\columnwidth]{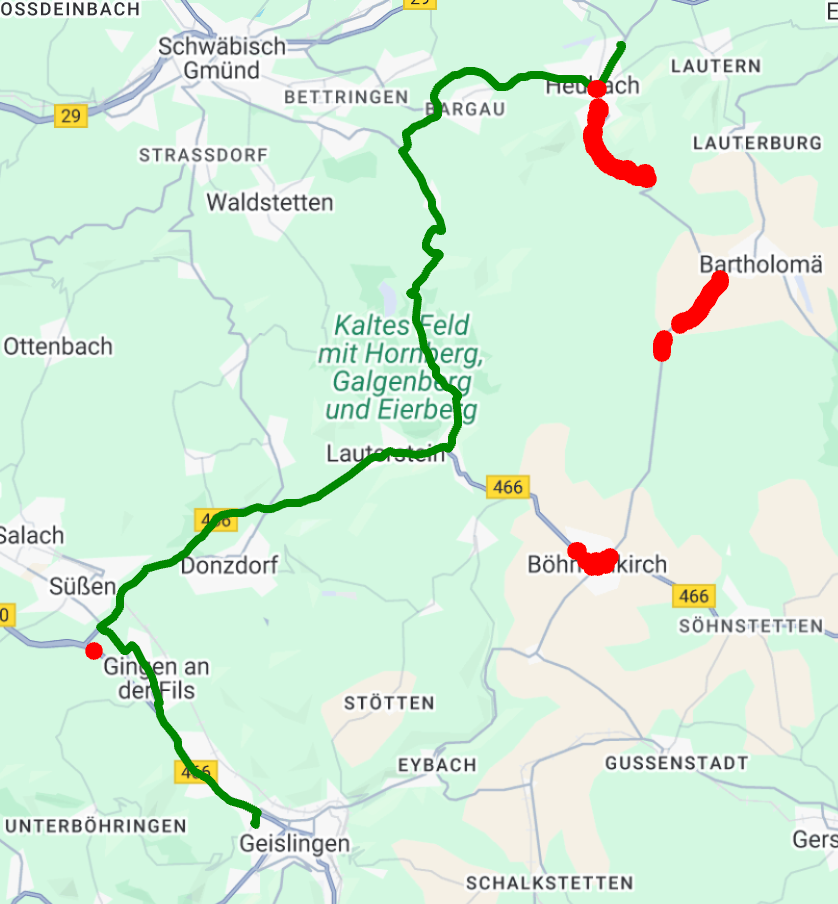}
    \caption{The minimizing example for Baden Medium. Sometimes, the SVE does contain segments outside of the closed arcs set.
    }
    \label{fig:worst-baden-medium}
\end{figure}

\subsection{Incident Scenarios}
\noindent 
{\bf Setup.}
We construct incident scenarios with $\gamma = 1.1$, $u(e) = 2\ell(e)$ for all $e\notin P_0\cup P_1\cup \hdots\cup P_{k-1}\cup P_k$, and $u(e) = y_e$ otherwise, for the $y_e$ defined in Section \ref{subsec:eval-setup}. Let $X$ denote the set of arcs $e\in E(G)$ for which $y_e \ne \ell(e)$. $X$ is called the \emph{penalized set} and $X = P_0\cup \hdots\cup P_{k-1}$.\\

\noindent 
{\bf Results.}
We evaluate SVEs on two metrics:

\begin{enumerate}
    \item The percentage of SVE segments that are also in the paths $P_0\cup P_1\cup\hdots\cup P_{k-1}$.
    \item The ratio between the support size of the SVE and $|P_0\cup \hdots\cup P_{k-1}|$.
\end{enumerate}

Both metrics capture smallness: the first tests precision of the SVE, while the second tests recall. We just report the minimum value of the first quantity for each dataset, as the minima are quite high. For the second quantity, we report percentiles. SVEs return a small subset (median around 10\%; see Table ~\ref{tab:result3}) of the penalties used to create the scenario and contain $>$83\% of relevant arcs for all datasets.

\begin{table}[H]
\centering
\begin{tabular}{|l|c|c|c|c|c|}
\hline
 & \textbf{\% val} & \textbf{\parbox{1.5cm}{\centering\% exp in paths}} & \multicolumn{3}{c|}{\parbox{2.5cm}{\centering\textbf{\# exp arcs / \# path arcs}}}\\ \hline
percentile &  & min & 50\% & 90\% & max \\ \hline
Seattle & 100\% & 87.8\% & 0.113 & 0.225 & 0.318 \\ \hline
WA Short & 94\% & 83.3\% & 0.094 & 0.299 & 0.434 \\ \hline
WA Med & 97\% & 95.7\% & 0.114 & 0.272 & 0.433 \\ \hline
WA Long & 86\% & 94.2\% & 0.111 & 0.209 & 0.343 \\ \hline
Ba Short & 96\% & 91.3\% & 0.07 & 0.348 & 0.925 \\ \hline
Ba Med & 98\% & 94.9\% & 0.11 & 0.246 & 0.421 \\ \hline
Ba Long & 98\% & 96.8\% & 0.111 & 0.24 & 0.381 \\ \hline
\end{tabular}
\vspace{0.05in}

\caption{Incident scenarios. SVEs are nearly a subset of the penalized arcs $X$, while being a small subset of those arcs. 
}
\label{tab:result3}
\end{table}

\section{Conclusion}\label{sec:conclusion}

Routing engines are used by billions of users per day to route in the presence of traffic, road closures, and other dynamic conditions. In this context, they often generate routes that deviate from the standard route for reasons that are difficult to discern. In this work, we introduced a linear programming-based technique for identifying the conditions that are relevant to a routing engine's choice of route. We gave theoretical and empirical justification for the effectiveness of these explanations.

This technique can be employed for many optimization problems that admit linear programming-based solutions. Examples include TSP, clustering, etc. Dynamic conditions are a central aspect to many of these optimization problems. 
This technique might be adaptable to these settings.

In this work, we focused on traffic-based scenarios, but our framework can apply to explain why routes optimizing other objectives, like scenic routes, safe routes, and minimum length routes, differ from the minimum duration route. 

\bibliographystyle{plain}
\bibliography{refs}

\begin{thebibliography}{10}

\bibitem{aigner2024framework}
Kevin-Martin Aigner, Marc Goerigk, Michael Hartisch, Frauke Liers, and Arthur Miehlich.
\newblock A framework for data-driven explainability in mathematical optimization.
\newblock In {\em Proceedings of the AAAI Conference on Artificial Intelligence}, volume~38, pages 20912--20920, 2024.

\bibitem{BGV89}
Michael~O. Ball, Bruce~L. Golden, and Rakesh~V. Vohra.
\newblock Finding the most vital arcs in a network.
\newblock {\em Oper. Res. Lett.}, 8(2):73--76, 1989.

\bibitem{CBB07}
Yanyan Chen, Michael G.~H. Bell, and Klaus Bogenberger.
\newblock Reliable pretrip multipath planning and dynamic adaptation for a centralized road navigation system.
\newblock {\em IEEE Transactions on Intelligent Transportation Systems}, 8(1):14--20, 2007.

\bibitem{CZ17}
Stephen~R. Chestnut and Rico Zenklusen.
\newblock Interdicting structured combinatorial optimization problems with {$\{0,1\}$}-objectives.
\newblock {\em Math. Oper. Res.}, 42(1):144--166, 2017.

\bibitem{dfmr20}
Sanjoy Dasgupta, Nave Frost, Michal Moshkovitz, and Cyrus Rashtchian.
\newblock Explainable k-means and k-medians clustering.
\newblock In {\em Proceedings of the 37th International Conference on Machine Learning}, ICML'20. JMLR.org, 2020.

\bibitem{doshi2017towards}
Finale Doshi-Velez and Been Kim.
\newblock Towards a rigorous science of interpretable machine learning.
\newblock {\em arXiv preprint arXiv:1702.08608}, 2017.

\bibitem{dwivedi-etal}
Rudresh Dwivedi, Devam Dave, Het Naik, Smiti Singhal, Rana Omer, Pankesh Patel, Bin Qian, Zhenyu Wen, Tejal Shah, Graham Morgan, and Rajiv Ranjan.
\newblock Explainable ai (xai): Core ideas, techniques, and solutions.
\newblock {\em ACM Comput. Surv.}, 55(9), January 2023.

\bibitem{erwig2024explanations}
Martin Erwig and Prashant Kumar.
\newblock Explanations for combinatorial optimization problems.
\newblock {\em Journal of Computer Languages}, 79:101272, 2024.

\bibitem{forel2023explainable}
Alexandre Forel, Axel Parmentier, and Thibaut Vidal.
\newblock Explainable data-driven optimization: from context to decision and back again.
\newblock In {\em International Conference on Machine Learning}, pages 10170--10187. PMLR, 2023.

\bibitem{forel2024don}
Alexandre Forel, Axel Parmentier, and Thibaut Vidal.
\newblock Don’t explain noise: Robust counterfactuals for randomized ensembles.
\newblock In {\em International Conference on the Integration of Constraint Programming, Artificial Intelligence, and Operations Research}, pages 293--309. Springer, 2024.

\bibitem{FH77}
D.~R. Fulkerson and Gary~C. Harding.
\newblock Maximizing the minimum source-sink path subject to a budget constraint.
\newblock {\em Math. Programming}, 13(1):116--118, 1977.

\bibitem{gpsy23}
A.~Gupta, M.~Pittu, O.~Svensson, and R.~Yuan.
\newblock The price of explainability for clustering.
\newblock In {\em 2023 IEEE 64th Annual Symposium on Foundations of Computer Science (FOCS)}, pages 1131--1148, Los Alamitos, CA, USA, nov 2023. IEEE Computer Society.

\bibitem{hoffman2018metrics}
Robert~R Hoffman, Shane~T Mueller, Gary Klein, and Jordan Litman.
\newblock Metrics for explainable ai: Challenges and prospects.
\newblock {\em arXiv preprint arXiv:1812.04608}, 2018.

\bibitem{IW02}
Eitan Israeli and R.~Kevin Wood.
\newblock Shortest-path network interdiction.
\newblock {\em Networks}, 40(2):97--111, 2002.

\bibitem{KBBEGRZ}
Leonid Khachiyan, Endre Boros, Konrad Borys, Khaled Elbassioni, Vladimir Gurvich, Gabor Rudolf, and Jihui Zhao.
\newblock On short paths interdiction problems: total and node-wise limited interdiction.
\newblock {\em Theory Comput. Syst.}, 43(2):204--233, 2008.

\bibitem{L17}
Euiwoong Lee.
\newblock {Improved Hardness for Cut, Interdiction, and Firefighter Problems}.
\newblock In Ioannis Chatzigiannakis, Piotr Indyk, Fabian Kuhn, and Anca Muscholl, editors, {\em 44th International Colloquium on Automata, Languages, and Programming (ICALP 2017)}, volume~80 of {\em Leibniz International Proceedings in Informatics (LIPIcs)}, pages 92:1--92:14, Dagstuhl, Germany, 2017. Schloss Dagstuhl -- Leibniz-Zentrum f{\"u}r Informatik.

\bibitem{lep20}
Patrick Lewis, Ethan Perez, Aleksandra Piktus, Fabio Petroni, Vladimir Karpukhin, Naman Goyal, Heinrich K\"{u}ttler, Mike Lewis, Wen-tau Yih, Tim Rockt\"{a}schel, Sebastian Riedel, and Douwe Kiela.
\newblock Retrieval-augmented generation for knowledge-intensive nlp tasks.
\newblock In H.~Larochelle, M.~Ranzato, R.~Hadsell, M.F. Balcan, and H.~Lin, editors, {\em Advances in Neural Information Processing Systems}, volume~33, pages 9459--9474. Curran Associates, Inc., 2020.

\bibitem{LundbergECDPNKH20}
Scott~M. Lundberg, Gabriel~G. Erion, Hugh Chen, Alex~J. DeGrave, Jordan~M. Prutkin, Bala Nair, Ronit Katz, Jonathan Himmelfarb, Nisha Bansal, and Su{-}In Lee.
\newblock From local explanations to global understanding with explainable {AI} for trees.
\newblock {\em Nat. Mach. Intell.}, 2(1):56--67, 2020.

\bibitem{ll17}
Scott~M. Lundberg and Su-In Lee.
\newblock A unified approach to interpreting model predictions.
\newblock In {\em Proceedings of the 31st International Conference on Neural Information Processing Systems}, NIPS'17, page 4768–4777, Red Hook, NY, USA, 2017. Curran Associates Inc.

\bibitem{molnar2020interpretable}
Christoph Molnar.
\newblock {\em Interpretable machine learning}.
\newblock Lulu. com, 2020.

\bibitem{moshkovitz2020explainable}
Michal Moshkovitz, Sanjoy Dasgupta, Cyrus Rashtchian, and Nave Frost.
\newblock Explainable k-means and k-medians clustering.
\newblock In {\em International conference on machine learning}, pages 7055--7065. PMLR, 2020.

\bibitem{ram23}
Ori Ram, Yoav Levine, Itay Dalmedigos, Dor Muhlgay, Amnon Shashua, Kevin Leyton-Brown, and Yoav Shoham.
\newblock In-context retrieval-augmented language models.
\newblock {\em Transactions of the Association for Computational Linguistics}, 11:1316--1331, 2023.

\bibitem{Samek-et-al}
Wojciech Samek, Gregoire Montavon, Andrea Vedaldi, Lars~Kai Hansen, and Klaus-Robert Muller.
\newblock {\em Explainable AI: Interpreting, Explaining and Visualizing Deep Learning}.
\newblock Springer Publishing Company, Incorporated, 1st edition, 2019.

\bibitem{setal24}
Weijia Shi, Sewon Min, Michihiro Yasunaga, Minjoon Seo, Richard James, Mike Lewis, Luke Zettlemoyer, and Wen-tau Yih.
\newblock {REPLUG}: Retrieval-augmented black-box language models.
\newblock In Kevin Duh, Helena Gomez, and Steven Bethard, editors, {\em Proceedings of the 2024 Conference of the North American Chapter of the Association for Computational Linguistics: Human Language Technologies (Volume 1: Long Papers)}, pages 8371--8384, Mexico City, Mexico, June 2024. Association for Computational Linguistics.

\bibitem{wotm24}
David~S. Watson, Joshua O'Hara, Niek Tax, Richard Mudd, and Ido Guy.
\newblock Explaining predictive uncertainty with information theoretic shapley values.
\newblock In {\em Proceedings of the 37th International Conference on Neural Information Processing Systems}, NIPS '23, Red Hook, NY, USA, 2024. Curran Associates Inc.

\bibitem{osmhighway}
OpenStreetMap Wiki.
\newblock Key:highway --- openstreetmap wiki{,}, 2024.
\newblock [Online; accessed 14-October-2024].

\end{thebibliography}

\appendix

\section{A Fast Combinatorial Algorithm}\label{app:algo}

The general idea of our combinatorial algorithm to find SVEs is to write the linear programming dual of the cut formulation, and then to consider the duality gap between the values of these two formulations. We show how a flow-based approach can drive this gap down to zero, giving us both the optimal primal and dual solutions. Recall the cut formulation:

\begin{align}
    \min \sum_{e \in E(G)} & \tau(e)(w_e - \ell(e)) && \tag{LP1} \label{LP1aa}\\
    w_e &\le u(e) && \forall e \in E(G) \label{eq:1a}\\
    w_e &\ge \ell(e) && \forall e\in E(G) \label{eq:2a} \\
    d_v - d_u &\leq w_e & &\forall e = (u,v) \in E(G) \label{eq:3a} \\
    d_v - d_u &= w_e && \forall e = (u,v) \in P. \label{eq:4a}
\end{align}

To begin, we first write the linear programming dual of the cut formulation. We call this dual program the \emph{flow formulation}, emphasizing the analogy between the flow/cut duality in network flows. It has a variable $a_e$ corresponding to each lower bound constraint, a variable $b_e$ corresponding to each upper bound constraint, and a variable $f_e$ corresponding with the two distance constraints.

\begin{align}
    \max \sum_{e \in E(G)} & (\ell(e)(a_e - \tau(e)) - u(e)b_e) && \tag{LP2} \\
    f_e &\ge 0 && \forall e \in E(G)\setminus P\\
    a_e,b_e &\ge 0 && \forall e \in E(G)\\
    \sum_{e\in \delta^+(v)} f_e &= \sum_{e\in \delta^-(v)} f_e && \forall v \in V\\
    a_e - b_e + f_e &= \tau(e) && \forall e \in E(G).
\end{align}

We now discuss some sanity checks related to this formulation. Note that if $u(e) < \ell(e)$, the objective value of this program is unbounded, as one can set $a_e = b_e$ to some arbitrarily large positive numbers. Furthermore, for any arc $e$, in an optimal solution, it is the case that either $a_e = 0$ or $b_e = 0$, as if both $a_e > 0$ and $b_e > 0$, the fact that $\ell(e) \le u(e)$ means that the objective value can be increased by decrementing both $a_e$ and $b_e$ by the same amount. 

Given feasible solutions $(w,d)$ and $(f,a,b)$ for the cut and flow formulations respectively, the \emph{duality gap} between the objective values of these two solutions can be written as
\begin{align*}
&\sum_e \tau(e)(w_e - \ell(e)) - \sum_e (\ell(e)(a_e - \tau(e)) - u(e)b_e) \\
&= \sum_{e=(u,v)} w_e(\tau(e) - (a_e - b_e + f_e)) \\
&+ \sum_e (w_e - \ell(e))a_e \\
&+ \sum_e (u(e) - w_e)b_e \\
&+ \sum_{e=(u,v)} (w_e - d_v + d_u)f_e\\
&+ \sum_{v\in V} d_v\left(\sum_{e\in \delta^-(v)} f_e - \sum_{e\in \delta^+(v)} f_e\right).
\end{align*}

Notice that this duality gap is always nonnegative when all constraints are satisfied. Our goal is to construct a pair of solutions for which the duality gap is zero, which ensures that both the primal and dual solutions are optimal. We now show how to get zero duality gap via an augmenting paths-like algorithm. For a graph $\wh{G} = (\wh{V},\wh{E})$ and vectors $\wh{\kappa}\in \mathbb{R}^{\wh{E}}$, $\wh{c}\in \mathbb{R}_{\ge 0}^{\wh{E}}$, and $\wh{W}\in \mathbb{R}$, define the following \emph{residual flow formulation} to facilitate the construction of augmenting paths, where $\wh{E}_{high}$ and $\wh{E}_{low}$ are the sets of arcs $e$ with $\wh{c}(e) = 0$ and $\wh{c}(e) > 0$ respectively, representing high and low flow (not capacity) respectively:

\begin{align}
    \max \wh{W} + \sum_e & \wh{\kappa}(e)\p f_e &&\\
    \p f_e &\ge 0 && \forall e \in \wh{E}\\
    \sum_{e\in \delta^+(v)} \p f_e &= \sum_{e\in \delta^-(v)} \p f_e && \forall v \in \wh{V}\\
    \p f_e &\le \wh{c}(e) && \forall e\in \wh{E}_{low}.
\end{align}

Only arcs in $\wh{E}_{low}$ are capacity constrained, as ones in $\wh{E}_{high}$ can be increased arbitrarily but at the cost of increasing $b_e$. Note that in this formulation, $\p f_e$ needs to be nonnegative on all arcs, not just the arcs outside of $P$. This is done because the residual problem models changes that can be made to the current flow formulation solution. $\wh{E}$ is also not necessarily just the set of original arcs $E$. It also includes \emph{residual arcs} that are reversals of arcs from the original flow formulation. See Figure \ref{fig:residual} for an algorithm that constructs the residual flow formulation associated with a solution to the flow formulation.

\begin{figure}
\colorbox{gray!20}{
\begin{minipage}{\textwidth}
\textbf{algorithm} \Residual($G=(V,E),\ell,u,f,a,b$):

\begin{enumerate}
\item If there exists an $e\in E$ for which both $a_e > 0$ and $b_e > 0$, throw an exception.
\item Let $\wh{E}_{res}$ denote the set of pairs $(v,u)\in V\times V$ for which $(u,v)\in E$ and $f_{(u,v)} > 0$. Let $\wh{E}_P$ denote the set of pairs $(v,u)\in V\times V$ for which $(u,v)\in P$.
\item Let $\wh{V} \gets V$ and $\wh{E} \gets E\cup \wh{E}_{res}\cup \wh{E}_P$. Initialize $F\gets \emptyset, \wh{F}_{res}\gets \emptyset, \wh{F}_P\gets \emptyset$.
\item Initialize $\wh{c}(e)\gets -\infty, \wh{\kappa}(e) = -\infty$, and $\wh{W}\gets 0$ and define them as follows. For each $e=(u,v)\in \wh{E}$,
    \begin{enumerate}
    \item If $e\in E$, do the following:
        \begin{enumerate}
        \item If $f_e < \tau(e)$, $a_e = \tau(e) - f_e > 0$ and $b_e = 0$, so increment $\wh{W}\gets \wh{W} - \ell(e)f_e$, let $\wh{\kappa}(e) \gets \max(\wh{\kappa}(e),-\ell(e))$, and update $\wh{c}(e)\gets \tau(e) - f_e$, add $e$ to $F$, and remove $e$ from $\wh{F}_{res}$ and $\wh{F}_P$ if $\wh{\kappa}(e)$ changed.
        \item Otherwise, $f_e \ge \tau(e)$, in which case $a_e = 0$ and $b_e = f_e - \tau(e)$, so increment $\wh{W}\gets \wh{W} + (u(e) - \ell(e))\tau(e) - u(e)f_e$, let $\wh{\kappa}(e) \gets \max(\wh{\kappa}(e),-u(e))$, and update $\wh{c}(e)\gets 0$, add $e$ to $F$, and remove $e$ from the other $F$-sets if $\wh{\kappa}(e)$ changed.
        \end{enumerate}
    \item If $e=(v,u)\in \wh{E}_{res}$, let $e'=(u,v)$. By definition, $e'\in E$. Do the following:
        \begin{enumerate}
        \item If $0 < f_{e'} \le \tau(e')$, let $\wh{\kappa}(e) \gets \max(\wh{\kappa}(e),\ell(e'))$ and update $\wh{c}(e)\gets f_{e'}$, add $e$ to $\wh{F}_{res}$, and remove $e$ from the other $F$-sets if $\wh{\kappa}(e)$ changed.
        \item If $f_{e'} > \tau(e')$, let $\wh{\kappa}(e) \gets \max(\wh{\kappa}(e),u(e'))$ and update $\wh{c}(e)\gets f_{e'} - \tau(e')$, add $e$ to $\wh{F}_{res}$, and remove $e$ from the other $F$-sets if $\wh{\kappa}(e)$ changed.
        \end{enumerate}
    \item If $e=(v,u)\in \wh{E}_P$, let $e'=(u,v)$. By definition, $e'\in P$. Do the following:
        \begin{enumerate}
        \item If $f_{e'} \le 0$, let $\wh{\kappa}(e) \gets \max(\wh{\kappa}(e),\ell(e'))$ and update $\wh{c}(e)\gets 0$, add $e$ to $\wh{F}_P$, and remove $e$ from the other $F$-sets if $\wh{\kappa}(e)$ changed.
        \end{enumerate}
    \end{enumerate}
\item \textbf{return} ($\wh{G}:=(\wh{V},\wh{E}),\wh{\kappa},\wh{c},\wh{W},F,\wh{F}_{res},\wh{F}_P$)
\end{enumerate}
\end{minipage}
}
\caption{Algorithm \Residual}
\label{fig:residual}
\end{figure}

We say that a solution $\p f$ to a residual flow formulation \Residual
given input $(G,\ell,u,f,a,b$) is \emph{nondegenerate} iff for any $(u,v)\in F$ for which both $(u,v)\in F$ and $(v,u)\in \wh{F}_{res}\cup \wh{F}_P$, $\partial f_{(u,v)} = 0$ or $\partial f_{(v,u)} = 0$. The following claim holds:

\begin{claim}\label{clm:nondegenerate}
Consider a flow formulation $(G,\ell,u)$ and a feasible solution $(f,a,b)$. Any feasible solution $\p f$ to \Residual($G,\ell,u,f,a,b$) can be converted in linear time to one with higher ($\ge$) objective value that is also nondegenerate.
\end{claim}

\begin{proof}
Consider any $e=(u,v)\in F$ for which $e'=(v,u)\in \wh{F}_{res}\cup \wh{F}_P$, $\p f_e > 0$, and $\p f_{e'} > 0$. Subtracting $\delta$ for some $0 < \delta \le \max(\p f_e, \p f_{e'})$ maintains feasibility and increases the objective value by $(\ell(e) - \ell(e))\delta = 0$ if $f_e < \tau(e)$, $(u(e) - \ell(e))\delta > 0$ if $f_e = \tau(e)$, and $(u(e) - u(e))\delta = 0$ if $f_e > \tau(e)$. Applying to all 2-cycles results in a solution that is nondegenerate and has higher objective value, as desired.
\end{proof}

\begin{figure}
\colorbox{gray!20}{
\begin{minipage}{\textwidth}
\textbf{algorithm}
\Modify($G=(V,E),\ell,u,f,a,b,\partial f$)\\

Note: $\p f$ is assumed to be a feasible nondegenerate solution to \Residual($G,\ell,u,f,a,b$). Furthermore, $F,\wh{F}_{res}$, and $\wh{F}_P$ form a partition of $\wh{E}$ by definition.\\

\begin{enumerate}
\item For each $e\in F$ with $\p f_e > 0$, let $f_e' \gets f_e + \p f_e$.
\item For each $e=(u,v)\in E$ for which $e'=(v,u)\in \wh{F}_{res}\cup \wh{F}_P$ and $\p f_{e'} > 0$, let $f_e' \gets f_e - \p f_{e'}$. By the nondegeneracy of $\p f$, this set of arcs $e$ is disjoint from the previous set of modified arcs $e$.
\item For all other $e\in E$, let $f_e' \gets f_e$.
\end{enumerate}

For all $e\in E$,
\begin{enumerate}
\item If $f_e' \le \tau(e)$, let $a_e' \gets \tau(e) - f_e'$ and $b_e' \gets 0$.
\item Otherwise, let $a_e' \gets 0$ and $b_e' \gets f_e' - \tau(e)$.
\end{enumerate}

\textbf{return} $(f',a',b')$.

\end{minipage}
}
\caption{Algorithm \Modify}
\end{figure}

We now show that a nondegenerate nonzero solution to the residual problem can be used to improve the current flow formulation solution to one with a higher objective value. The resulting objective value is equal to the objective value of the residual solution:

\begin{claim}\label{clm:feasible}
Consider a flow formulation $(G,\ell,u)$ and a feasible solution $(f,a,b)$. For any feasible nondegenerate solution $\p f$ to the residual flow formulation \Residual($G,\ell,u,f,a,b$), the output $(f',a',b')$ of the algorithm \Modify($G,\ell,u,f,a,b,\p f$) is a feasible solution to the flow formulation $(G,\ell,u)$ with the same objective value as the solution $\p f$ to the residual flow formulation.
\end{claim}

\begin{proof}
We check feasibility and objective matching together. First, notice that $\wh{W}$ is equal to the flow formulation objective value of $f$, so we just have to check that the residual flow formulation objective value of $\p f$ is equal to the difference between the flow formulation objective values of $f$ and $f'$. Feasibility is guaranteed in each step by the choice of $\wh{c}$. We now elaborate on this intuition.

For each $e\in E$, we start by verifying that $f_e' \ge 0$. If $e\in F$ and $\p f_e > 0$, then $f_e' = f_e + \p f_e$. If $f_e < \tau(e)$, then $\wh{c}(e) = \tau(e) - f_e$ and $f_e' = f_e + \p f_e \le f_e + (\tau(e) - f_e) = \tau(e)$, so $f_e'$ lies in the range for which $b_e' = 0$, $a_e' = a_e - \p f_e$, and $e$ contributes a flow formulation objective value difference of $\ell(e)(a_e' - a_e) = -\ell(e)\p f_e = \wh{\kappa}(e)\p f_e$, as desired. If $f_e \ge \tau(e)$, then $f_e'$ always lies in the range for which $a_e' = 0$, $b_e' = b_e + \p f_e$, and $e$ contributes a flow formulation objective value difference of $-u(e)(b_e' - b_e) = -u(e)\p f(e) = \wh{\kappa}\p f(e)$ as desired. This completes the case where $e\in F$.

Now, suppose that $e=(u,v)$ has the property that $e'=(v,u)\in \wh{F}_{res} \cup \wh{F}_P$ and $\p f_{e'} > 0$. If $f_e > \tau(e)$, then $\wh{c}(e') = f_e - \tau(e) > 0$, $e\in E_{low}$, and $f_e' = f_e - \p f_{e'} \ge f_e - \wh{c}(e') = \tau(e')$, so $f_e'$ and $f_e$ lie in the range for which $a_e' = a_e = 0$, and $e$ contributes a flow formulation objective value difference of $-u(e)(b_e' - b_e) = u(e)\p f_{e'} = \wh{\kappa}(e')\p f_{e'}$ as desired. If $0 < f_e \le \tau(e)$, then $\wh{c}(e) = f_e > 0$, $e\in E_{low}$, and $f_e' = f_e - \p f_{e'} \ge f_e - \wh{c}(e') = 0$, so $f_e'$ and $f_e$ lie in the range for which $b_e' = b_e = 0$, and $e$ contributes a flow formulation objective value difference of $\ell(e)(a_e' - a_e) = \ell(e)(\tau(e) - f_e' - (\tau(e) - f_e)) = \ell(e)\p f_{e'} = \wh{\kappa}(e')\p f_{e'}$ as desired. If $f_e \le 0$, then $e\in P$, which means that feasibility in the range with $\wh{\kappa}(e) = \ell(e')$ range is guaranteed and the desired objective change is achieved. This completes all cases in which $f_e' \ne f_e$ and all other cases have no flow formulation objective change contribution, as desired.
\end{proof}

Thus, whenever we can find a feasible residual solution with better objective value, it can be used to improve the current flow formulation objective. Each residual problem solution is a circulation. Circulations can always be written as sums of cycles. In particular, whenever the residual graph contains a cycle with positive $\wh{\kappa}$-weight, the objective value can be increased. When the residual graph does not contain such a cycle, we show that the algorithm \CutCertificate produces a cut formulation solution with duality gap 0, certifying that the current flow formulation solution is optimal.

\begin{figure}
\colorbox{gray!20}{
\begin{minipage}{\textwidth}
\CutCertificate($G=(V,E),P,\ell,u,f,a,b$)\\

Note: we assume that the residual graph $\wh{G}$ of \Residual($G,\ell,u,f,a,b$) has no $\wh{\kappa}$-weighted positive cycles. (Otherwise, \Residual has a nonzero optimal solution and calling \Modify will produce a better solution to the flow formulation.)

\begin{enumerate}
\item Let $s$ and $t$ denote the origin and destination of the path $P$. By assumption, the graph $\wh{G} = (\wh{V}, \wh{E})$ has no $-\wh{\kappa}$-weighted negative cycles, so distances from $s$ are well-defined (non-negative-infinite).
\item Run Bellman-Ford from $s$. For each $v\in \wh{V}$, let $d_v$ denote the $-\wh{\kappa}$-weighted distance from $s$ to $v$ in $\wh{G}$.
\item For each $e\in E\setminus P$ with $f_e = 0$, let $w_e = \ell(e)$. For all other $e = (u,v)\in E$, let $w_e = d_v - d_u$.
\end{enumerate}
\textbf{return} $(w,d)$.

\end{minipage}
}
\caption{Algorithm \CutCertificate}
\end{figure}

We now show that the algorithm \CutCertificate returns the desired cut solution:

\begin{claim}\label{clm:optimal}
For a graph $G$, path $P$, weight functions $\ell$ and $u$, and a feasible flow formulation solution $(f,a,b)$, $(w,d) \gets $ \CutCertificate($G,P,\ell,u,f,a,b$) is a feasible cut formulation solution for which the pair of solutions $(f,a,b),(w,d)$ has duality gap 0.
\end{claim}

\begin{proof}

\textbf{Feasibility: $d_v - d_u \le w_e$ and $d_v - d_u = w_e$}: For every arc $e = (u,v)\in E\setminus P$ with $f_e = 0$, $e\in \wh{E}$ and $\wh{\kappa}(e) \ge -\ell(e)$. Thus, by the triangle inequality, $d_v \le d_u - \wh{\kappa}(e) \le d_u + \ell(e) = d_u + w_e$, as desired. For every other arc $e = (u,v)\in E$, $d_v - d_u = w_e$ by definition of $w_e$, as desired.

\textbf{Feasibility: $w_e \ge \ell(e)$ and $w_e \le u(e)$}: For every arc $e = (u,v)\in E\setminus P$ with $f_e = 0$, $w_e = \ell(e) \le u(e)$ by definition, as desired. For every other arc $e = (u,v)\in E$, $e' = (v,u)\in \wh{E}$ because either $e\in P$, in which case $e'\in \wh{E}_P$, or $e\notin P$, in which case $f_e > 0$ and $e'\in \wh{E}_{res}$. Furthermore, $\wh{\kappa}(e) \ge -u(e)$ and $\wh{\kappa}(e') \ge \ell(e)$ for all values of $f_e$. Thus, applying the triangle inequality to both arcs shows that $d_v - d_u \le -\wh{\kappa}(e) \le u(e)$ and $d_u - d_v \le -\wh{\kappa}(e') \le -\ell(e)$. Combining these inequalities and using the fact that $d_v - d_u = w_e$ shows that $\ell(e) \le w_e \le u(e)$, as desired.

\textbf{Duality Gap}: The duality gap is the sum of five terms, as written earlier. Feasibility of the flow formulation solution $(f,a,b)$ guarantees that the first and last terms are equal to 0. The fourth term is equal to 0 because in all cases where $w_e - d_v + d_u = 0$ for an arc $e = (u,v)$, $f_e = 0$. For every arc $e = (u,v)\in E\setminus P$ with $f_e = 0$, $w_e = \ell(e)$ and $b_e = 0$, so the second and third terms of the duality gap are also 0, as desired. For every other arc $e = (u,v)\in E$, $e' = (v,u)\in \wh{E}$ also, as discussed earlier. If $f_e < \tau(e)$, $b_e = 0$, $\wh{\kappa}(e) \ge -\ell(e)$, and $\wh{\kappa}(e') \ge \ell(e)$. Therefore, by the triangle inequality, $w_e = d_v - d_u \le -\wh{\kappa}(e) \le \ell(e)$ and $-w_e = d_u - d_v \le -\wh{\kappa}(e') \le -\ell(e)$, so $w_e = \ell(e)$, as desired. If $f_e = \tau(e)$, then both $a_e = 0$ and $b_e = 0$, as desired. If $f_e > \tau(e)$, $a_e = 0$, $\wh{\kappa}(e) \ge -u(e)$, and $\wh{\kappa}(e') \ge u(e)$. Therefore, by the triangle inequality, $w_e = d_v - d_u \le -\wh{\kappa}(e) \le u(e)$ and $-w_e = d_u - d_v \le -\wh{\kappa}(e') \le -u(e)$, so $w_e = u(e)$, as desired. Thus, in all cases, all five terms are 0, as desired.
\end{proof}

\end{document}